\documentclass[11pt,reqno]{article}

\usepackage{microtype}

\usepackage{amsmath, amssymb, amsfonts, bm, cite}\usepackage{mathrsfs}
\renewcommand{\mathbf}[1]{\boldsymbol{#1}}

\usepackage{geometry}
\usepackage{fullpage}
\usepackage{pdfsync, graphicx}
\usepackage{amsthm}

\usepackage[colorlinks=true,citecolor=blue,bookmarksnumbered=true,hyperfootnotes=false]{hyperref}

\newcommand{\concept}[1]{\emph{#1}}
\newcommand{\hyper}[1]{\mathcal{#1}}
\newcommand{\RHtree}[2]{{\mathbb{T}}_{#1,#2}}

\newtheorem{theorem}{Theorem}[section]

\newtheorem{lemma}[theorem]{Lemma}
\newtheorem{corollary}[theorem]{Corollary}
\newtheorem{definition}{Definition}[section]
\newtheorem{remark}{Remark}[section]
\newtheorem*{remark*}{Remark}

\theoremstyle{plain}
\newtheorem{proposition}[theorem]{Proposition}

\newcommand{\group}[1]{{\mathbb{#1}}}
\newcommand{\mat}[1]{\boldsymbol{#1}}
\newcommand{\TSAW}{\mathcal{T}_{\mathrm{SAW}}}
\newcommand{\ptree}[2]{\mathbb{P}_{#1}^{#2}} 
 
\newcommand{\pgraphhv}[1]{p_{{\hyper{H}},{v}}^{#1}} 
\newcommand{\ptreehv}[1]{\mathbb{P}_{\hyper{T}}^{#1}} 
\newcommand{\VHtree}[2]{\widehat{\mathbb{T}}_{#1,#2}}
\newcommand{\defeq}{\triangleq}
\newcommand{\ceil}[1]{\ensuremath{\left\lceil#1\right\rceil}}
\newcommand{\floor}[1]{\ensuremath{\left\lfloor#1\right\rfloor}}
\newcommand{\neigh}[2]{B_{#1 #2}}
\newcommand{\localto}{\to_{\mathrm{loc}}}
\newcommand{\law}[1]{\mathbb{#1}}

\title{Counting hypergraph matchings up to uniqueness threshold}

\author{
Renjie Song\thanks{Department of Computer Science and Technology, Nanjing University, China. \texttt{song.renjie@foxmail.com}. Supported by NSFC grants 61272081 and 61321491.}
\and
Yitong Yin\thanks{State Key Laboratory for Novel Software Technology, Nanjing University, China. \texttt{yinyt@nju.edu.cn}. Supported by NSFC grants 61272081 and 61321491.}
\and
Jinman Zhao\thanks{Department of Computer Science, University of Wisconsin-Madison. \texttt{jinman.zhao@gmail.com}. This work was done when Jinman Zhao was an undergraduate student at Nanjing University.}
}

\date{}

\begin{document}

\maketitle
\begin{abstract}
We study the problem of approximately counting matchings in hypergraphs of bounded maximum degree and maximum size of hyperedges. With an activity parameter $\lambda$, each matching $M$ is assigned a weight $\lambda^{|M|}$. The counting problem is formulated as computing a partition function that gives the sum of the weights of all matchings in a hypergraph.
This problem unifies two extensively studied statistical physics models in approximate counting: the hardcore model (graph independent sets) and the monomer-dimer model (graph matchings).

For this model, the critical activity $\lambda_c= \frac{d^d}{k (d-1)^{d+1}}$ is the threshold for the uniqueness of Gibbs measures on the infinite $(d+1)$-uniform $(k+1)$-regular hypertree. 
Consider hypergraphs of maximum degree at most $k+1$ and maximum  size of hyperedges at most $d+1$.
We show that when $\lambda < \lambda_c$, there is an FPTAS for computing the partition function; and when $\lambda = \lambda_c$, there is a PTAS for computing the log-partition function.
These algorithms are based on the decay of correlation (strong spatial mixing) property of Gibbs distributions. 
When $\lambda > 2\lambda_c$, there is no PRAS for the partition function or the log-partition function unless NP$=$RP.

Towards obtaining a sharp transition of computational complexity of approximate counting, we study the local convergence from a sequence of finite hypergraphs to the infinite lattice with specified symmetry. We show a surprising connection between the local convergence and the reversibility of a natural random walk. This leads us to a barrier for the hardness result: The non-uniqueness of infinite Gibbs measure is not realizable by any finite gadgets.

\end{abstract}

\section{Introduction}

Counting problems have long been studied in the context of statistical physics models.  
Perhaps the two most well studied statistical physics models for approximate counting are the \concept{hardcore model} and the \concept{monomer-dimer model}.

In the hardcore model, given a graph $G=(V,E)$ and a vertex-activity $\lambda$, the model assigns each independent set $I$ of $G$ a weight $w^{\mathsf{IS}}_\lambda(I)=\lambda^{|I|}$. A natural probability distribution, the Gibbs distribution, is defined over all independent sets of $G$ as $\mu^{\mathsf{IS}}_\lambda(I)={w^{\mathsf{IS}}_\lambda(I)}/{Z^{\mathsf{IS}}_\lambda(G)}$ where the normalizing factor $Z^{\mathsf{IS}}_\lambda(G)=\sum_{I}w^{\mathsf{IS}}_\lambda(I)$  
is the \concept{partition function}. 
In the monomer-dimer model, given a graph $G=(V,E)$ and an edge-activity $\lambda$, the model assigns each matching $M$ of $G$ a weight $w^{\mathsf{M}}_\lambda(M)=\lambda^{|M|}$. The Gibbs distribution over all matchings of $G$ is defined accordingly. And the partition function now becomes $Z^{\mathsf{M}}_\lambda(G)=\sum_{M}w^{\mathsf{M}}_\lambda(M)$. 
The counting problems are then formulated as computing the partition functions $Z^{\mathsf{IS}}_\lambda(G)$ and $Z^{\mathsf{M}}_\lambda(G)$, or the \concept{log-partition functions} $\log Z^{\mathsf{IS}}_\lambda(G)$ and $\log Z^{\mathsf{M}}_\lambda(G)$.

It was well known that the hardcore model exhibits the following phase transition. For the infinite $(d+1)$-regular tree $\mathbb{T}_d$, there is a critical activity $\lambda_c(\mathbb{T}_d)=d^d/(d-1)^{d+1}$,  called the \concept{uniqueness threshold}, such that when $\lambda<\lambda_c$ the correlation between the marginal distribution at the root and any boundary condition on leaves at level $t$ decays exponentially in the depth $t$, but when $\lambda>\lambda_c$ the boundary-to-root correlation remains substantial even as $t\to\infty$. This property of correlation decay is also called spatial mixing, and was known to be equivalent to the uniqueness of the infinite-volume Gibbs measure on the infinite $(d+1)$-regular tree $\mathbb{T}_d$~\cite{weitz2005combinatorial}.
In a seminal work~\cite{weitz2006counting}, Weitz showed that for all $\lambda<\lambda_c(\mathbb{T}_d)$ the decay of correlation holds for the hardcore model on all graphs of maximum degree bounded by $d+1$ and there is a deterministic FPTAS for approximately computing the partition function on all such graphs. Here the specific notion of decay of correlation established is the strong spatial mixing. 
The connection of approximability of partition function to the phase transition of the model is further strengthened in a series of works~\cite{sly2010computational, sly2014counting, galanis2012improved, galanis2012inapproximability} which show that unless NP=RP there is no PRAS for the partition function or the log-partition function of the hardcore model when $\lambda>\lambda_c(\mathbb{T}_d)$ on graphs with maximum degree bounded by $d+1$.

For the monomer-dimer model, it was well known that the model has no such phase transition~\cite{heilmann1972existence, heilmann1972theory}. And analogously there is an FPRAS due to Jerrum and Sinclair~\cite{jerrum1989approximating} for the partition function of the monomer-dimer model on all graphs. In~\cite{bayati2007simple} strong spatial mixing with an exponential rate was established for the model on all graphs with maximum degree bounded by an arbitrary constant and a deterministic FPTAS was also given for the partition function on all such graphs.

In this paper, we study \concept{hypergraph matchings}, a model that unifies both the hardcore model and the monomer-dimer model.
A hypergraph $\hyper{H}=(V,E)$ consists of a vertex set $V$ and a collection $E$ of vertex subsets, called the (hyper)edges. A \concept{matching} of $\hyper{H}$ is a set $M\subseteq E$ of disjoint hyperedges in $\hyper{H}$.
Given a hypergraph $\hyper{H}$ and an \concept{activity parameter} $\lambda>0$, a \concept{configuration} is a matching $M$ of $\hyper{H}$, and is assigned a weight $w_{\lambda}(M)=\lambda^{|M|}$. The \concept{Gibbs measure} over all matchings of $\hyper{H}$ is defined as $\mu(M)={w_\lambda(M)}/{Z_{\lambda}(\hyper{H})}$, where the normalizing factor $Z_{\lambda}(\hyper{H})$ is the \concept{partition function} for the model, defined as:
\[
Z_{\lambda}(\hyper{H})=\sum_{M:\text{  matching of }\hyper{H}} \lambda^{|M|}.
\]
This model represents an interesting subclass of Boolean CSP defined by the matching (packing) constraints. It also unifies the hardcore model and the monomer-dimer model.
Consider the family of hypergraphs of maximum edge size $d+1$ and maximum degree $k+1$:
\begin{itemize}
\item When $d=1$, the model becomes the monomer-dimer model on graphs of maximum degree $k+1$. 

\item When $k=1$, the partition function takes sum over independent sets in the dual graph, and the model becomes the hardcore model on graphs of maximum degree $d+1$.
\end{itemize}
For hypergraphs, the study of approximate counting hypergraph matchings was initiated in~\cite{karpinski2013approximate}. In~\cite{dudek2014approximate}, an FPTAS was obtained for counting matchings in 3-uniform hypergraphs of maximum degree at most 3 by considering the correlation decay for the independent sets in claw-free graphs. In~\cite{liu2013fptas}, an FPTAS was given for 3-uniform hypergraphs of maximum degree at most 4 by the correlation decay of the original CSP. All these results assumed $\lambda=1$, i.e.~the problem of counting the number of matchings in a hypergraph.

\paragraph*{Our results.}
We show that for hypergraph matchings $\lambda_c=\lambda_c(\RHtree{d}{k})=\frac{d^{d}}{k(d-1)^{{d+1}}}$ is the uniqueness threshold on the infinite $(d+1)$-uniform $(k+1)$-regular hypertree~$\RHtree{d}{k}$.

\begin{proposition}\label{proposition-uniqueness-threshold}

There is a unique Gibbs measure on matchings of $\RHtree{d}{k}$ if and only if $\lambda\le\lambda_c$.
\end{proposition}

This fact was implicit in the literature. Here we give a formal proof.
It subsumes the well-known uniqueness threshold $\lambda_c(\RHtree{d}{1})=\frac{d^{d}}{(d-1)^{{d+1}}}$ for the hardcore model on the infinite $(d+1)$-regular tree and also the lack of phase-transition for the monomer-dimer model.

We then establish the decay of correlation for hypergraph matchings on all hypergraphs with bounded maximum size of hyperedges and bounded maximum degree when the activity $\lambda$ is in the uniqueness regime for the uniform regular hypertree.
The specific notion of decay of correlations that we establish here is the strong spatial mixing~\cite{weitz2006counting} (see Section~\ref{sec:prelim} for a formal definition).
Consequently, we give an FPTAS for the partition function when $\lambda$ is in the interior of the uniqueness regime, and a PTAS for the log-partition function when $\lambda$ is at the critical threshold.

\begin{theorem}\label{theorem-main}
For every finite integers $d,k\ge 1$, the following holds for matchings with activity $\lambda$ on all hypergraphs of maximum edge-size at most $d+1$ and maximum degree at most $k+1$:
\begin{itemize}
\item if $\lambda < \lambda_c$, the model exhibits strong spatial mixing at an exponential rate and there exists an FPTAS for computing the partition function;
\item if $\lambda = \lambda_c$, the model exhibits strong spatial mixing at a polynomial rate and there is a PTAS for computing the log-partition function.
\end{itemize}
\end{theorem}

\begin{remark*}
The theorem unifies the strong spatial mixing and FPTAS for the hardcore model~\cite{weitz2006counting} and the monomer-dimer model~\cite{bayati2007simple}, and also covers as special cases the results for approximate counting non-weighted hypergraph matchings in~\cite{karpinski2013approximate, dudek2014approximate, liu2013fptas}.

For hypergraph matchings, the case of critical threshold is of significance. There is a natural combinatorial problem that corresponds to the threshold case: counting matchings in $3$-uniform hypergraphs of maximum degree at most 5. Here $d=2$, $k=4$, and the critical $\lambda_c=\frac{d^d}{k(d-1)^{d+1}}=1$, which corresponds to counting the number of hypergraph matchings without weight.
\end{remark*}

Unlike most recent correlation-decay-based algorithms, where the strong spatial mixings were established by a potential analysis,  we do not use the potential method to analyze the decay of correlation. Instead, we prove the following stronger extremal statement.

\begin{proposition}\label{proposition-ssm-tree}
For hypergraph matchings, the worst case of (weak or strong) spatial mixing, in terms of decay rate, among all hypergraphs of maximum edge-size at most $d+1$ and maximum degree at most $k+1$, is represented by the weak spatial mixing on $\RHtree{d}{k}$.
\end{proposition}
We construct a hypergraph version of Weitz's self-avoiding walk tree. Then we show that weak spatial mixing on the uniform regular hypertree implies strong spatial mixing on all smaller hypertrees by a step-by-step comparison of correlation decay. This was the original approach used by Weitz for the hardcore model~\cite{weitz2006counting}. Compared to the more recent potential method~\cite{li2013correlation,li2012approximate,restrepo2013improved,sinclair2014approximation,sinclair2013spatial,liu2013fptas}, this method of analyzing the decay of correlation has the advantage in dealing with the critical case.

On the other hand, due to a simple reduction from the inapproximability of the hardcore model in the non-uniqueness regime~\cite{sly2014counting}, we have the following hardness result.

\begin{theorem}\label{theorem-main-hardness}
If $\lambda > \frac{2k+1+(-1)^k}{k+1}\lambda_c\approx2\lambda_c$, then there is no PRAS for the partition function or the log-partition function for the family of hypergraphs stated in Theorem~\ref{theorem-main}, unless NP=RP.
\end{theorem}

\begin{figure}
\centering
\includegraphics[scale=.6]{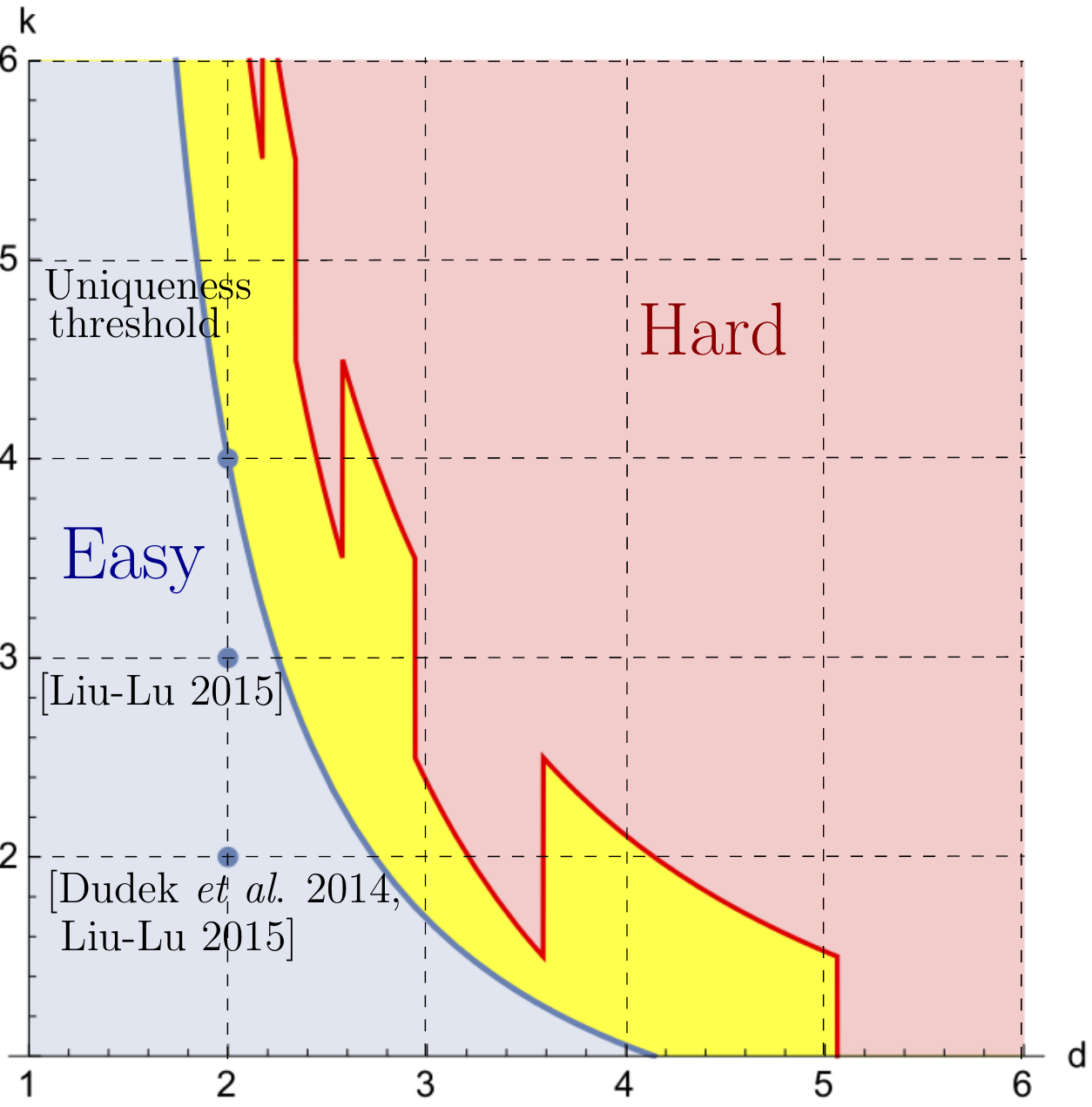}
\caption{\footnotesize{The classification of computational complexity of approximately counting matchings in hypergraphs of max-degree $(k+1)$ and max-edge-size $(d+1)$ when $\lambda=1$. The blue curve is the uniqueness threshold. The non-continuity of the red curve is due to rounding.}}
\label{fig:classification}
\end{figure}

Figure~\ref{fig:classification} illustrates the classification of approximability of counting hypergraph matchings when $\lambda=1$. Each integral point $(d,k)$ corresponds to the problem of approximately counting matchings in hypergraphs of max-degree $(k+1)$ and max-edge-size $(d+1)$. The landscape will continuously change when $\lambda$ changes.

It is worth noticing that in our reduction the hard instances contain many small cycles, while from the algorithmic side the worst cases for the decay of correlation are trees. This obvious inconsistency between upper and lower bounds and the \textit{ad hoc} nature of the simple reduction seem to suggest that the current hardness threshold is not optimal.

We then explore the possibility of bringing the current hardness threshold from $\approx2\lambda_c$ down to the phase-transition threshold $\lambda_c$. 
We discover a reason why getting the exact transition of approximability could be so challenging for this model on hypergraphs.

To state our discovery, let us first review the current approach for establishing computational phase transition for approximate counting~\cite{dyer2002counting, mossel2009hardness, sly2010computational, sly2014counting, galanis2012improved, galanis2012inapproximability,galanis2014inapproximability}, which consists of two main steps:
\begin{itemize}
\item (from all infinite measures to finitely many infinite measures)
The uniqueness threshold $\lambda_c(\mathbb{T}_d)$ for the Gibbs measure on the infinite regular tree $\mathbb{T}_d$ is achieved by a sub-family of Gibbs measures with simple structure: the Gibbs measures that are invariant under a group $\group{G}$ of automorphisms on $\mathbb{T}_d$. For the hardcore model, these are the so-called \concept{semi-translation invariant} Gibbs measures, which are invariant under parity-preserving automorphisms on $\mathbb{T}_d$, and the threshold $\lambda_c(\mathbb{T}_d)$ for the uniqueness of all Gibbs measures on $\mathbb{T}_d$ is the same as the threshold $\lambda_c(\mathbb{T}_d^{\group{G}})$ for the uniqueness of only those Gibbs measures that are invariant under the group $\group{G}$ of parity-preserving automorphisms. 
\item (from finitely many infinite measures to finite measures)
A sequence of (possibly random) finite graphs $G_n$ is constructed to converge locally to $\mathbb{T}_d^\group{G}$, the infinite tree $\mathbb{T}_d$ equipped with the symmetry specified by group $\group{G}$. For the hardcore model, and more generally antiferromagnetic spin systems, $G_n$ are the random regular bipartite graphs~\cite{dyer2002counting, mossel2009hardness, sly2010computational, sly2014counting, galanis2012improved, galanis2012inapproximability,galanis2014inapproximability}, which converge locally to the infinite tree $\mathbb{T}_d$ respecting the symmetry between vertices of the same parity. The ``random''  and ``regular'' parts in this construction guarantee to preserve the local tree structure in distribution, while the bipartiteness respects the parity of vertices.
\end{itemize}
For the model of hypergraph matchings, the first step follows. We show that there indeed is a group $\widehat{\group{G}}$ of automorphisms on the infinite $(d+1)$-uniform $(k+1)$-regular hypertree $\RHtree{d}{k}$ such that $\lambda_c(\RHtree{d}{k})=\lambda_c(\RHtree{d}{k}^{\widehat{\group{G}}})$, i.e.~the uniqueness of Gibbs measure on $\RHtree{d}{k}$ is represented precisely by the uniqueness of only those Gibbs measures invariant under $\widehat{\group{G}}$. This gives a natural generalization of semi-translation Gibbs measures to the hypergraph model.

However, we show that there does not exist \emph{any} sequence of (deterministic or random) finite hypergraphs that converge locally to $\RHtree{d}{k}^{\widehat{\group{G}}}$ unless $k=1$ where the model degenerates to the hardcore model on graphs. In fact, we give a complete characterization of the symmetry described by a group ${\group{G}}$ of automorphisms on $\RHtree{d}{k}$ that there exists a sequence of finite hypergraphs that converge locally to $\RHtree{d}{k}^{\group{G}}$.

\begin{theorem}\label{theorem-main-local-convergence}
Let $\group{G}$ be a group of automorphisms on $\RHtree{d}{k}$ with finitely many orbits.
There exist a sequence of random finite hypergraphs $\hyper{H}_n$ that converge locally to $\RHtree{d}{k}^{\group{G}}$
if and only if the uniform random walk on $\RHtree{d}{k}$ projected onto the orbits of $\group{G}$ is reversible.
\end{theorem}

See Theorem~\ref{theorem-local-converge} and its proof for more details of Theorem~\ref{theorem-main-local-convergence}.

\paragraph*{Discussion.} To summarize our discoveries for the model of hypergraph matchings:  
\begin{itemize}
\item
Theorem~\ref{theorem-main} implicitly but rigorously shows that the worst case for the decay of correlation among a family of hypergraphs with bounded maximum degree and bounded maximum edge-size, is achieved by the infinite uniform regular hypertree. 
\item
However, in the current inapproximability stated by Theorem~\ref{theorem-main-hardness}, the hard instances are not locally tree-like, but rather, the gadgets locally converge to an infinite hypergraph which is not a hypertree (see Section~\ref{sec:algorithms}).
\item
And finally, Theorem~\ref{theorem-main-local-convergence} gives an explanation of this inconsistence between upper and lower bounds: the extremal case for the decay of correlation in Theorem~\ref{theorem-main}, which is achieved by an infinite-hypertree measure, can never be realized by \emph{any} finite hypergraphs.\footnote{In fact, aided by numerical simulations, so far we have not encountered any family of measures on the infinite uniform regular hypertree $\RHtree{d}{k}$ realizable by finite hypergraphs, whose uniqueness threshold is below $2\lambda_c$. This seems to provide some empirical evidence for that on finite hypergraphs, the worst case for uniqueness might not be locally tree-like.}
\end{itemize}
Altogether, these discoveries deliver the following very interesting message:
In order to establish a sharp connection between computational complexity of approximate counting and phase transitions for hypergraph matchings or other more general models, a more fine-grained definition of uniqueness on finite graphs is necessary.

\paragraph*{Remark on exposition.}
For convenience of visualizing the results, all our results in the rest of the paper are presented for \emph{independent sets} in the dual hypergraphs.
Note that matchings are equivalent to independent sets under hypergraph duality. The only effect of duality on a family of hypergraphs with bounded maximum edge size and bounded maximum degree is to switch the bounds on the edge size and the degree.
We emphasize that our notion of hypergraph independent set is different from the more popular definition used in~\cite{bezakova2015counting, bordewich2008path}. We call a vertex subset $I\subseteq V$ in a hypergraph $\hyper{H}=(V,E)$ an independent set if no two vertices in $I$ are contained in the same hyperedge, while in~\cite{bezakova2015counting, bordewich2008path}, an $I\subseteq V$ is an independent set if it does not contain any hyperedge as subset.

\paragraph*{Related works.}
Approximate counting of hypergraph matchings was studied in~\cite{karpinski2013approximate} for hypergraphs with restrictive structures, and in~\cite{liu2013fptas, dudek2014approximate} for hypergraphs with bounded edge size and maximum degree. In~\cite{bordewich2008path, lu2015fptas}, approximate counting of a variant of hypergraph independent sets was studied, where the definition of hypergraph independent set is different from ours.
In a very recent breakthrough~\cite{bezakova2015counting}, FPTAS for this problem is obtained when there is no strong spatial mixing. In~\cite{galanis2015complexity}, the hardness is established for a class of hypergraph models including ours.

The spatial mixing (decay of correlation) is already a widely studied topic in Computer Science, because it may support FPTAS for \#P-hard counting problems. The decay of correlation was established via the self-avoiding walk tree for the hardcore model~\cite{weitz2006counting, sinclair2013spatial}, monomer-dimer model~\cite{bayati2007simple, sinclair2014approximation}, and two-spin systems~\cite{li2013correlation, li2012approximate, sinclair2014approximation}.
Similar tree-structured recursions were employed to prove the decay of correlation for multi-spin systems~\cite{gamarnik2012correlation, lu2013improved, gamarnik2013strong} and more general CSPs~\cite{lin2014simple, lu2014fptas, liu2013fptas}.

\section{Preliminaries}
\label{sec:prelim}
For a \concept{hypergraph} $\hyper{H}=(V,E)$, the \concept{size} of a hyperedge $e\in E$ is its cardinality $|e|$, and the \concept{degree} of a vertex $v\in V$,  denoted by $\deg{v}=\mathrm{deg}_{\hyper{H}}(v)$,  is the number of hyperedges $e\in E$ \concept{incident to} $v$, i.e.~satisfying $v\in e$.
A hypergraph $\hyper{H}$ is \concept{$k$-uniform} if all hyperedges are of the same size $k$, and is \concept{$d$-regular} if all vertices have the same degree $d$.
The \concept{incidence graph} of a hypergraph $\hyper{H}=(V,E)$ is a bipartite graph with $V$ and $E$ as vertex sets on the two sides, such that each $(v,e)\in V\times E$ is a bipartite edge if and only if $v$ is incident to $e$.

A \concept{matching} of hypergraph $\hyper{H}=(V,E)$ is a set $M\subseteq E$ of disjoint hyperedges in $\hyper{H}$. Given an \concept{activity parameter} $\lambda>0$, the \concept{Gibbs measure} is a probability distribution over matchings of $\hyper{H}$ proportional to the weight $w^{\mathsf{M}}_\lambda(M)=\lambda^{|M|}$, defined as $\mu^{\mathsf{M}}_\lambda(M)=w^{\mathsf{M}}_\lambda(M)/Z^{\mathsf{M}}_\lambda(\hyper{H})$, where the normalizing factor $Z^{\mathsf{M}}_\lambda(\hyper{H})=\sum_{M}w^{\mathsf{M}}_{\lambda}(M)$ is the partition function.

Similarly, an \concept{independent set} of hypergraph $\hyper{H}=(V,E)$ is a set $I\subseteq V$ of vertices satisfying $|I\cap e|\le 1$ for all hyperedges $e$ in $\hyper{H}$. The \concept{Gibbs measure} over independent sets of $\hyper{H}$ with activity $\lambda>0$ is given by
\begin{align}\label{IS-measure}
\mu^{\mathsf{IS}}_\lambda(I)=\frac{w^{\mathsf{IS}}_\lambda(I)}{Z^{\mathsf{IS}}_\lambda(\hyper{H})}=\frac{\lambda^{|I|}}{Z^{\mathsf{IS}}_\lambda(\hyper{H})},
\end{align}
where the normalizing factor $Z^{\mathsf{IS}}_\lambda(\hyper{H})=\sum_{I}w^{\mathsf{IS}}_{\lambda}(I)$ is the partition function for independent sets of $\hyper{H}$ with activity $\lambda$.

Independent sets and matchings are equivalent under hypergraph duality. The \concept{dual} of a hypergraph $\hyper{H}=(V,E)$, denoted by $\hyper{H}^*=(E^*,V^*)$, is the hypergraph whose vertex set is denoted by $E^*$ and edge set is denoted by $V^*$, such that every vertex $v\in V$ (and every hyperedge $e\in E$) in $\hyper{H}$ is one-to-one corresponding to a hyperedge $v^*\in V^*$ (and a vertex $e^*\in E^*$), such that $e^*\in v^*$ if and only if $v\in e$.
Note that under duality, matchings and hypergraphs are the same CSP and hence result in the same Gibbs measure, which remains to be true even with activity $\lambda$. Also a family of hypergraphs of bounded maximum edge size and bounded maximum degree is transformed under duality to a family of hypergraphs with the bounds on the edge size and degree exchanged.

\begin{remark}\label{remark-duality}
With the above equivalence under duality, from now on we state all our results in terms of the independent sets in the dual hypergraph and omit the superscript $\cdot^{\mathsf{IS}}$ in notations. 
\end{remark}

Given the Gibbs measure over independent sets of hypergraph $\hyper{H}$ and a vertex $v$, we define the \concept{marginal probability} $p_{v}$ as 
\[
p_{v}
=p_{\hyper{H},v} 
=\Pr[v\in I ]
\]
which is the probability that $v$ is in an independent set $I$ sampled from the Gibbs measure (such a vertex is also said to be \emph{occupied}). Given a vertex set $\Lambda\subset V$, a \concept{configuration} is a $\sigma_\Lambda\in\{0,1\}^\Lambda$ which corresponds to an independent set $I_\Lambda$ partially specified over $\Lambda$ such that $\sigma_\Lambda(v)$ indicates whether a $v\in\Lambda$ is occupied by the independent set.
We further define the marginal probability $p_{\hyper{H},v}^{\sigma_\Lambda} $ as
\[
p_{v}^{\sigma_\Lambda}
=p_{\hyper{H},v}^{\sigma_\Lambda}
=\Pr[v\in I \mid I_\Lambda=\sigma_{\Lambda}]
\]
which is the probability that $v$ is occupied under the Gibbs measure conditioning on the configuration of vertices in $\Lambda\subset V$ being fixed as $\sigma_\Lambda$.

\begin{definition}
The independent sets of a finite hypergraph $\hyper{H}=(V,E)$ with activity $\lambda>0$ exhibit \concept{weak spatial mixing (WSM)} with rate $\delta : \mathbb{N} \to \mathbb{R}^+$ if for any $v \in V$, $\Lambda \subseteq V$, and any two configurations $\sigma_\Lambda, \tau_\Lambda\in\{0,1\}^{\Lambda}$ which correspond to two independent sets partially specified on $\Lambda$,
\[
\left| p_{v}^{\sigma_\Lambda} - p_{v}^{\tau_\Lambda} \right|
\leq \delta(\mathrm{dist}_{\hyper{H}}(v, \Lambda)),
\]
where $\mathrm{dist}_{\hyper{H}}(v, \Lambda)$ is the shortest distance between $v$ and any vertex in $\Lambda$ in hypergraph $\hyper{H}$.
\end{definition}

\begin{definition}
The independent sets of a finite hypergraph $\hyper{H}=(V,E)$ with activity $\lambda>0$ exhibit \concept{strong spatial mixing (SSM)} with rate $\delta : \mathbb{N} \to \mathbb{R}^+$ if for any $v \in V$, $\Lambda \subseteq V$, and any two configurations $\sigma_\Lambda, \tau_\Lambda\in\{0,1\}^{\Lambda}$ which correspond to two independent sets partially specified on $\Lambda$,
\[
\left| p_{v}^{\sigma_\Lambda} - p_{v}^{\tau_\Lambda} \right|
\leq \delta(\mathrm{dist}_{\hyper{H}}(v, \Delta)),
\]
where $\Delta \subseteq \Lambda$ stands for the subset on which $\sigma_\Lambda$ and $\tau_\Lambda$ differ and $\mathrm{dist}_{\hyper{H}}(v, \Delta)$ is the shortest distance between $v$ and any vertex in $\Delta$ in hypergraph $\hyper{H}$.
\end{definition}

The definitions of WSM and SSM extend to infinite hypergraphs with the same conditions to be satisfied for every finite region $\Psi\subset V$ conditioning on the vertices in $\partial\Psi$ being unoccupied.

\section{Gibbs measures on the infinite tree}\label{sec:uniqueness}
We follow Remark~\ref{remark-duality} and state our discoveries in terms of independent sets in the dual hypergraphs.
Let $\RHtree{k}{d}$ be the infinite $(k+1)$-uniform $(d+1)$-regular hypertree, whose incidence graph is the infinite tree in which all vertices with parity 0 are of degree $(k+1)$ and all vertices with parity 1 are of degree $(d+1)$.
A probability measure $\mu$ on hypergraph independent sets of $\RHtree{k}{d}$ is \concept{Gibbs} if for any finite sub-hypertree $\hyper{T}$, conditioning $\mu$ upon the event that all vertices on the outer boundary of $\hyper{T}$ are unoccupied gives the same distribution on independent sets of $\hyper{T}$ as defined by~\eqref{IS-measure} with $\hyper{H}=\hyper{T}$.
We further consider the \concept{simple} Gibbs measures satisfying \concept{conditional independence}: Conditioning $\mu$ on a configuration of a subset $\Lambda$ of vertices results in a measure in which the configurations on the components separated by $\Lambda$ are independent of each other. The Gibbs distribution on a finite hypergraph is always simple.
A Gibbs measure  on $\RHtree{k}{d}$ is \concept{translation-invariant} if it is invariant under all automorphisms of $\RHtree{k}{d}$.
Fix an automorphism group $\group{G}$ of $\RHtree{k}{d}$.
A \concept{$\group{G}$-translation-invariant} Gibbs measure on $\RHtree{k}{d}$ is a measure that is invariant under all automorphisms from $\group{G}$.
For example, the \concept{semi-translation-invariant} Gibbs measures on regular tree are invariant under all parity-preserving automorphisms on $\RHtree{1}{d}$.
The natural group actions of $\group{G}$ respectively on vertices and hyperedges partition the sets of vertices and hyperedges into orbits. For example, in the semi-translation-invariant symmetry on regular tree, vertices with the same parity form an orbit.
We will show that $\lambda_c(\RHtree{k}{d})=\frac{d^d}{k(d-1)^{d+1}}$ is the uniqueness threshold for the Gibbs measures on hypergraph independent sets of $\RHtree{k}{d}$. Furthermore,  this uniqueness threshold is achieved by a family of Gibbs measures with simple structure.

\begin{theorem}\label{thm-uniqueness-threshold}
There is always a unique simple translation-invariant Gibbs measure on independent sets of $\RHtree{k}{d}$. Let $\lambda_c=\lambda_c(\RHtree{k}{d})=\frac{d^d}{k(d-1)^{d+1}}$. There is a unique Gibbs measure on $\RHtree{k}{d}$ if and only if $\lambda \le \lambda_c$.
Furthermore, there is an automorphism group $\widehat{\group{G}}$ on $\RHtree{k}{d}$ which classifies all vertices of $\RHtree{k}{d}$ into 2 orbits, such that the threshold for the uniqueness of $\widehat{\group{G}}$-translation invariant Gibbs measures on $\RHtree{k}{d}$, denoted as $\lambda_c(\RHtree{k}{d}^{\widehat{\group{G}}})$, is $\lambda_c(\RHtree{k}{d}^{\widehat{\group{G}}})=\lambda_c(\RHtree{k}{d})$.

\end{theorem}
This proves the uniqueness threshold stated in Proposition~\ref{proposition-uniqueness-threshold}.

\subsection{Branching matrices}

The automorphism group ${\group{G}}$ on $\RHtree{k}{d}$ can be described conveniently by a notion of {branching matrices}.
For an automorphism group ${\group{G}}$ on $\RHtree{k}{d}$, the natural group actions of $\group{G}$ respectively on vertices and hyperedges partition the sets of vertices and hyperedges into orbits. Let $\tau_v$ and $\tau_e$ be the respective numbers of orbits for vertices and hyperedges. For each $i\in[\tau_v]$, we say a vertex is of \concept{type-$i$} if it is in the $i$-th orbit for vertices; and the same also applies to hyperedges.
Assuming the symmetry on $\RHtree{k}{d}$ given by automorphism group $\group{G}$, the \concept{hypergraph branching matrices}, or just \concept{branching matrices}, are the following two nonnegative integral matrices:
\[
\boldsymbol{D}=\boldsymbol{D}^{\tau_v\times \tau_e}=[d_{ij}]\quad\mbox{and}\quad \boldsymbol{K}=\boldsymbol{K}^{\tau_e\times \tau_v}=[k_{ji}],
\]
which satisfy that for any $i\in[\tau_v]$ and $j\in[\tau_e]$:
\begin{itemize}
\item every vertex in $\RHtree{k}{d}$ of type-$i$ is incident to precisely $d_{ij}$ hyperedges of type-$j$;
\item every hyperedge in $\RHtree{k}{d}$ of type-$j$ contains precisely $k_{ji}$ vertices of type-$i$.
\end{itemize}
The $\boldsymbol{D}$ and $\boldsymbol{K}$ are transition matrices from vertex-types to hyperedge-types and vice versa in $\RHtree{k}{d}$. The definition can be seen as a hypergraph generalization of the branching matrix for multi-type Galton-Watson tree~\cite{restrepo2013improved}.
Since types (orbits) are invariant under all automorphisms from $\group{G}$, it is clear that the above $\boldsymbol{D}$ and $\boldsymbol{K}$ are well-defined for every automorphism group $\group{G}$ on $\RHtree{k}{d}$ with finitely many orbits. 

\begin{proposition}\label{prop-bracnching-matrix}
Every automorphism group $\group{G}$ on $\RHtree{k}{d}$ with finitely many orbits can be identified by a pair of branching matrices $\boldsymbol{D}$ and $\boldsymbol{K}$ with rules as described above and satisfy: (1) $\sum_{j}d_{ij}=d+1$ and $\sum_{i}k_{ji}=k+1$; (2) $d_{ij}=0$ if and only if $k_{ji}=0$; and (3) $\boldsymbol{D}\boldsymbol{K}$ and $\boldsymbol{K}\boldsymbol{D}$ are irreducible.

Conversely, any pair of nonnegative integral matrices $\boldsymbol{D}$ and $\boldsymbol{K}$ satisfying these conditions are branching matrices for some automorphism group $\group{G}$ on $\RHtree{k}{d}$.

\end{proposition}
\begin{proof}
Let $\group{G}$ be an automorphism group on $\RHtree{k}{d}$ with finitely many orbits. It is trivial to see that the branching matrices $\boldsymbol{D}$ and $\boldsymbol{K}$ are well-defined and satisfy $\sum_{j}d_{ij}=d+1$ and $\sum_{i}k_{ji}=k+1$.

A vertex $v$ of type-$i$ is incident to a hyperedge $e$ of type-$j$ if and only if $e$ of type-$j$ contains a vertex $v$ of type $i$, thus $k_{ji} \ne 0$ if and only if $d_{ij} \ne 0$.

The irreducibility of $\boldsymbol{D}\boldsymbol{K}$ and $\boldsymbol{K}\boldsymbol{D}$ follows that of  the matrix $\begin{bmatrix}
\boldsymbol{0} & \boldsymbol{D}\\
\boldsymbol{K} & \boldsymbol{0}
\end{bmatrix}$, which is a consequence to the that every type of vertex and hyperedge is accessible from all other types of vertices and hyperedges, which follows the simple fact that the incidence graph $\RHtree{k}{d}$ is strongly connected.

Conversely, let $\mat{D}$ and $\mat{K}$ be a pair of nonnegative integral matrices satisfying the conditions above. We can start from any vertex  (or hyperedges) $o$ of type-$i$ and construct an infinite hypertree rooted at $o$ with each vertex and hyperedge labeled with the respective type according to the rules specified by the branching matrices $\mat{D}$ and $\mat{K}$. Since $d_{ij}=0$ if and only if $k_{ji}=0$, the construction is always possible.
Since $\sum_{j}d_{ij}=d+1$ and $\sum_{i}k_{ji}=k+1$, the resulting infinite hypertree must be $k$-uniform and $d$-regular.
Since $\mat{D}\mat{K}$ and $\mat{K}\mat{D}$ are irreducible, no matter how we choose the type for the root $o$, the resulting hypertree contains all  types of vertices and hyperedges.

We can then construct an automorphism group $\group{G}$ on $\RHtree{k}{d}$ according with orbits being the types just specified.
For every pair of vertices (or hyperedges) $u,v$ with the same type, by generating the hypertree according to $\mat{D}$, $\mat{K}$ starting from $u$ and $v$ respectively, we obtain an automorphism $\phi_{u \to v}$ on $\RHtree{k}{d}$ which maps $u$ to $v$ and preserves the types of all vertices and hyperedges.
Let $\group{G}= \langle \, \{\phi_{u \to v} \mid \forall u, v \text{ with the same type}\} \, \rangle$ be the group generated from all such automorphisms. Then $\mat{D}$ and $\mat{K}$ are branching matrices for automorphism group $\group{G}$ on $\RHtree{k}{d}$.
\end{proof}

\subsection{Extremal Gibbs measures}\label{subsection-extremal}
Consider a special automorphism group $\widehat{\group{G}}$ on $\RHtree{k}{d}$ defined by the following branching matrices $(\widehat{\boldsymbol{D}},\widehat{\boldsymbol{K}})$.
Assume that there are two vertex-types and two hyperedge-types, both denoted as $\{+,-\}$, and the branching matrices are defined as $\widehat{\boldsymbol{D}}=\begin{bmatrix}
1 & d\\
d & 1
\end{bmatrix}$ and $\widehat{\boldsymbol{K}}=\begin{bmatrix}
k & 1\\
1 & k
\end{bmatrix}$, i.e.:
\begin{enumerate}
\item
every `$\pm$'-vertex is incident to a `$\pm$'-hyperedge and $d$ `$\mp$'-hyperedges;
\item
every `$\pm$'-hyperedge contains $k$ `$\pm$'-vertices and a `$\mp$'-vertex.
\end{enumerate}
See Figure~\ref{fig:coloring} for an illustration.
\begin{figure}
\centering
\includegraphics[scale=.6]{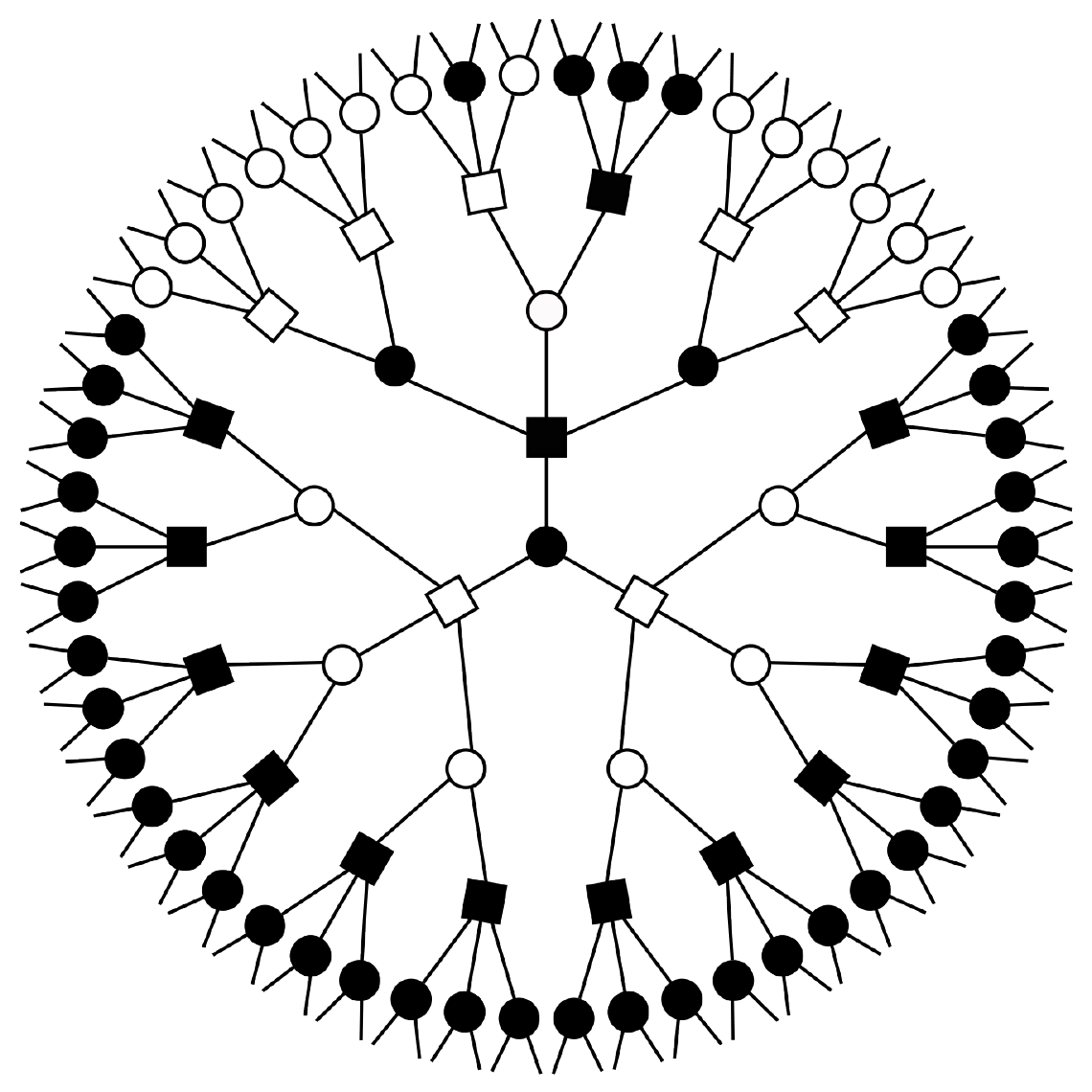}
\caption{\footnotesize{Classifying vertices and hyperedges of $\RHtree{3}{2}$ into two types `$+$'(black) and `$-$'(white). The hypergraph is represented as its incidence graph where circles stand for vertices and squares stand for hyperedges.}}
\label{fig:coloring}
\end{figure} 
Fix a `$+$'-vertex $v$ in $\RHtree{k}{d}$ as the root. Let $\mu^+$ (resp.~$\mu^-$) be the Gibbs measure on $\RHtree{k}{d}$ defined by conditioning on all vertices to be occupied for the $t$-th `$+$'-vertices (resp.~`$-$'-vertices) along all path from the root and taking the weak limit as $t\to\infty$. Note that for the 2-coloring given by $\widehat{\boldsymbol{D}}$ and $\widehat{\boldsymbol{K}}$, on any path any `$\pm$'-vertex has a `$\mp$'-vertex within 2 steps, so the limiting sequence is well-defined. And by symmetry, starting from a root of type-`$-$' gives the same pair of measures.

The $\mu^{\pm}$ generalize the extremal semi-translation-invariant Gibbs measures on infinite regular trees.
For hypertree $\RHtree{k}{d}$ with $k\ge 2$, there are no parity-preserving automorphisms. Nevertheless, the symmetry given by $\widehat{\boldsymbol{D}}$ and $\widehat{\boldsymbol{K}}$ generalizes the parity-preserving automorphisms to hypertrees and has the similar phase-transition as semi-translation-invariant Gibbs measures on trees.

The $\mu^{\pm}$ are simple and are $\widehat{\group{G}}$-translation-invariant for the automorphism group $\widehat{\group{G}}$ with orbits given by $\widehat{\boldsymbol{D}}$ and $\widehat{\boldsymbol{K}}$. In fact, they are extremal $\widehat{\group{G}}$-translation-invariant Gibbs measures on $\RHtree{k}{d}$. We will see that the model has uniqueness if and only if $\mu^+=\mu^-$. 

\subsection{Uniqueness of Gibbs measures}
\begin{lemma}
Let $\mu$ be a simple Gibbs measure on independent sets of $\RHtree{k}{d}$.
Let $v$ be a vertex in $\RHtree{k}{d}$ and $v_{ij}$ the $j$-th vertex (besides $v$) in the $i$-th hyperedge incident to $v$, for $i=1,2,\ldots,d+1$ and $j=1,2,\ldots, k$. Let $p_v=\mu[\,v\text{ is occupied}\,]$ and $p_{v_{ij}}=\mu[\,v_{ij}\text{ is occupied}\,]$. It holds that
\begin{align}\label{eq:regular-tree-marginal}
p_v = \lambda (1 - p_v)^{-d} \prod_{i=1}^{d+1} \left( 1 - p_v - \sum_{j=1}^{k} p_{v_{ij}} \right).
\end{align}
\end{lemma}
\begin{proof}
Since $\mu$ is a Gibbs measure, for any vertex $v$ in $\RHtree{k}{d}$, it holds that
\[
p_v=\mu[\,v\text{ is occupied}\,] = \frac{\lambda}{1+\lambda} \cdot \mu[\,\text{all the neighbors of }v\text{ are unoccupied}\,]
\]
On the other hand, since $\mu$ is simple, conditioning on the root being unoccupied the sub-hypertrees are independent of each other, thus
\begin{align*}
&\mu[\,\text{all the neighbors of }v\text{ are unoccupied}\,] \\
=& \mu[\,v\text{ is occupied}\,]  \cdot \mu[\,\text{all the neighbors of }v\text{ are unoccupied}\mid v\text{ is occupied}\,] \\
&+ \mu[\,v\text{ is unoccupied}\,] \prod_{i=1}^{d+1} \mu[\,\forall 1\le j\le k, v_{ij}\text{ is unoccupied}\mid v\text{ is unoccupied}\,] \\
=& p_v+ (1-p_v) \prod_{i=1}^{d+1} \left(1 - \sum_{j=1}^k \mu[\,v_{ij}\text{ is occupied} \mid v\text{ is unoccupied}\,]\right).
\end{align*}
Note that for any two adjacent vertices $v,v_{ij}$, we have $\mu[\,v_{ij}\text{ is occupied}\,] = \mu[\,v_{ij}\text{ is occupied} \mid v\text{ is unoccupied}\,] \cdot \mu[\,v\text{ is unoccupied}\,] $, thus
\[
\mu[\,v_{ij}\text{ is occupied} \mid v\text{ is unoccupied}\,]=\frac{\mu[\,v_{ij}\text{ is occupied}\,]}{1-\mu[\,v\text{ is occupied}\,]}=\frac{p_{v_{ij}}}{1-p_v}.
\]
The lemma follows by combining everything together.
\end{proof}

Equation~\eqref{eq:regular-tree-marginal} gives an infinite system involving all vertices in $\RHtree{k}{d}$.
If the simple Gibbs measure $\mu$ is $\group{G}$-translation-invariant for some automorphism group $\group{G}$ on $\RHtree{k}{d}$, the marginal probability $p_v=\mu[\,v\text{ is occupied}\,]$ depends only on the type (orbit) of $v$.
\begin{corollary}\label{corollary-regular-tree-marginal-DK}
Let $\mu$ be a simple $\group{G}$-translation-invariant Gibbs measure on $\RHtree{k}{d}$ with branching matrices $\boldsymbol{D}^{\tau_v\times \tau_e}=[d_{ij}]$ and $\boldsymbol{K}^{\tau_e\times \tau_v}=[k_{ji}]$.
For every $i\in[\tau_{v}]$, let $p_i=\mu[\,v\text{ is occupied}\,]$ for vertex $v$ in $\RHtree{k}{d}$ of type-$i$. It holds for every $s\in[\tau_v]$ that
\begin{align*}
p_s = \lambda (1 - p_s)^{-d} \prod_{j\in[\tau_e]}\left( 1 - \sum_{i\in[\tau_v]} k_{ji} \cdot p_{i} \right)^{d_{ij}}.
\end{align*}
\end{corollary}

Applying with the branching matrices $\widehat{\boldsymbol{D}}$ and $\widehat{\boldsymbol{K}}$ defined in Section~\ref{subsection-extremal}, the system in Corollary~\ref{corollary-regular-tree-marginal-DK} becomes
\begin{equation*}
\begin{cases}
p_+ = \lambda (1 - p_+)^{-d} (1 - k  \, p_+ - p_-) (1 - p_+ - k \,  p_-)^d, \\
p_- = \lambda (1 - p_-)^{-d} (1 - k  \, p_- - p_+) (1 - p_- - k \,  p_+)^d.
\end{cases}
\end{equation*}
Let $x=\frac{kp_+}{1 - p_- - k \, p_+}$ and $y=\frac{kp_-}{1 - p_+ - k \, p_-}$. The system becomes
$\begin{cases}
y=f(x)\\ 
x=f(y) 
\end{cases}$, where $f(x)=\frac{k\lambda}{(1+x)^d}$ is the hardcore tree-recursion.
Since $f(x)$ is positive and decreasing in $x$, it follows that there is a unique positive $\hat{x}$ such that $\hat{x}=f(\hat{x})$, which means there is always a unique simple translation-invariant Gibbs measure on $\RHtree{k}{d}$. It is well-known (see~\cite{galanis2012inapproximability} and~\cite{kelly1985stochastic,spitzer1975markov}) the system has three distinct  solutions $(\hat{x},\hat{x}),(x^+,x^-)$ and $(x^-,x^+)$ where $0<x^-<\hat{x}<x^+$, when $k\lambda>d^d/(d-1)^{d+1}$, i.e.~$\lambda>\lambda_c(\RHtree{k}{d})=\frac{d^d}{k(d-1)^{d+1}}$; and the three solutions collide into a unique solution $(\hat{x},\hat{x})$ when $\lambda\le \lambda_c(\RHtree{k}{d})$, which means there is a unique simple $\widehat{\group{G}}$-translation-invariant Gibbs measure on $\RHtree{k}{d}$ if and only if $\lambda\le\lambda_c(\RHtree{k}{d})$. 
Recall that $\mu^\pm$ are simple and are extremal $\widehat{\group{G}}$-translation-invariant Gibbs measures, and hence it also holds that $\mu^+=\mu^-$ if and only if $\lambda\le \lambda_c(\RHtree{k}{d})$, therefore, it holds that $\lambda_c(\RHtree{k}{d})=\lambda_c(\RHtree{k}{d}^{\widehat{\group{G}}})$.
In particular if $\lambda> \lambda_c(\RHtree{k}{d})$, then $\mu^+\neq\mu^-$ and the Gibbs measure on $\RHtree{k}{d}$ is non-unique.

To complete the proof of Theorem~\ref{thm-uniqueness-threshold}, we only need to show the Gibbs measure on $\RHtree{k}{d}$ is unique if $\lambda \le \lambda_c(\RHtree{k}{d})$. This is implied by the weak spatial mixing on $\RHtree{k}{d}$ when $\lambda\le\lambda_c$, proved later in Theorem~\ref{theorem-wsm-IS}. 
With the weak spatial mixing on $\RHtree{k}{d}$, the uniqueness of the Gibbs measure is implied by a generic equivalence between weak spatial mixing and uniqueness of Gibbs measure (see e.g.~\cite{weitz2005combinatorial}).

\section{The hypergraph self-avoiding walk tree}
\label{sec:SAW}

We call a hypergraph a hypertree if its incidence graph has no cycles.
Let $\hyper{T}=(V,E)$ be a rooted hypertree with vertex $v$ as its root.
We assume that root $v$ is incident to $d$ distinct hyperedges $e_1,e_2,\ldots,e_d$, such that for $i=1,2,\ldots,d$,

\begin{itemize}
\item $|e_i|=k_i+1$; and
\item $e_i=\{v,v_{i1},v_{i2},\ldots,v_{ik_i}\}$.
\end{itemize}
For $1\le i\le d$ and $1\le j\le k_i$, let $\hyper{T}_{ij}$ be the sub-hypertree rooted at $v_{ij}$. Recall that all hypertrees considered by us satisfy the property that any two hyperedges share at most one common vertex, thus all $v_{ij}$ are distinct and the sub-hypertrees $\hyper{T}_{ij}$ are disjoint.

Let $\Lambda\subset V$. Let $\sigma_\Lambda\in\{0,1\}^{\Lambda}$ be a configuration indicating an independent set partially specified on vertex set $\Lambda$, and for each $1\le i\le d$ and $1\le j\le k_i$, let $\sigma_{\Lambda_{ij}}$ be the restriction of $\sigma_\Lambda$ on the sub-hypertree $\hyper{T}_{ij}$.
Consider the ratios of marginal probabilities:

\begin{align*}
R^{\sigma_\Lambda}_{\hyper{T}}={p_{\hyper{T},v}^{\sigma_\Lambda}}/{\left(1-p_{\hyper{T},v}^{\sigma_\Lambda}\right)}\quad\text{ and }\quad
R^{\sigma_{\Lambda_{ij}}}_{\hyper{T}_{ij}}={p_{\hyper{T}_{ij},v_{ij}}^{\sigma_{\Lambda_{ij}}}}/{\left(1-p_{\hyper{T}_{ij},v_{ij}}^{\sigma_{\Lambda_{ij}}}\right)}.
\end{align*}
The following recursion can be easily verified due to the disjointness between sub-hypertrees:
\begin{align}
R^{\sigma_\Lambda}_{\hyper{T}}
=
\lambda\prod_{i=1}^d\frac{1}{1+\sum_{j=1}^{k_i}R^{\sigma_{\Lambda_{ij}}}_{\hyper{T}_{ij}} }.\label{eq:tree-recursion}
\end{align}
This is the ``tree recursion'' for hypergraph independent sets. The tree recursions for the hardcore model~\cite{weitz2006counting} and the monomer-dimer model~\cite{bayati2007simple} can both be interpreted as special cases.

For general hypergraphs which are not trees, we construct a hypergraph version of \concept{self-avoiding-walk tree}, which allows computing  marginal probabilities in arbitrary hypergraphs with the tree recursion. Moreover, we show that the uniform regular hypertree is the worst case for SSM among all hypergraphs of bounded maximum edge-size and bounded maximum degree.

\begin{theorem}\label{thm-saw}
For any positive integers $k,d$ and any positive $\lambda$, if the independent sets of $\RHtree{k}{d}$ with activity $\lambda$ exhibit strong spatial mixing with rate $\delta(\cdot)$, then the independent sets of any hypergraph of maximum edge size at most $(k+1)$ and maximum degree at most $(d+1)$, with activity $\lambda$,  exhibit strong spatial mixing with the same rate $\delta(\cdot)$.
\end{theorem}

Under duality, the same holds for the hypergraph matchings.

We then define the hypergraph self-avoiding walk tree.
A walk in a hypergraph $\hyper{H}=(V,E)$ 
is a sequence $(v_0,e_1,v_1,\ldots,e_{\ell},v_\ell)$ of alternating vertices and hyperedges such that every two consecutive vertices $v_{i-1},v_{i}$ are incident to the hyperedge $e_{i}$ between them. A walk $w=(v_0,e_1,v_1,\ldots,e_{\ell},v_\ell)$ is called \concept{self-avoiding} if:
\begin{itemize}
\item $w=(v_0,e_1,v_1,\ldots,e_{\ell},v_\ell)$ forms a simple path in the incidence graph of $\hyper{H}$; and
\item for every $i=1,2,\ldots,\ell$, vertex $v_i$ is incident to none of $\{e_1,e_2,\ldots,e_{i-1}\}$.
\end{itemize}
Note that the second requirement is new to the hypergraphs.

A self-avoiding walk $w=(v_0,e_1,v_1,\ldots,e_{\ell},v_\ell)$ can be extended to a \concept{cycle-closing} walk $w'=(v_0,e_1,v_1,\ldots,e_{\ell},v_\ell,e',v')$ so that  the suffix $(v_{i},e_{i+1},v_{i+1},\ldots,e_{\ell},v_\ell,e',v')$, for some $0\le i\le \ell-1$, of the walk forms a simple cycle in the incidence graph of $\hyper{H}$.
We call $v'$ the cycle-closing vertex.

Given a hypergraph $\hyper{H}=(V,E)$, an ordering of incident hyperedges at every vertex can be arbitrarily fixed, so that for any two hyperedges $e_1,e_2$ incident to a vertex $u$ we use $e_1<_ue_2$ to denote that $e_1$ is ranked higher than $e_2$ according to the ordering of hyperedges incident to $u$.
With this local ordering of hyperedges, given any vertex $v\in V$, a rooted hypertree $\hyper{T}=\TSAW(\hyper{H}, v)$, called the \concept{self-avoiding walk (SAW) tree}, is constructed as follows:
\begin{enumerate}
\item
Every vertex of $\hyper{T}$ corresponds to a distinct {self-avoiding walk} in $\hyper{H}$ originating from $v$,  where the root corresponds to the trivial walk $(v)$.

\item For any vertex $u$ in $\hyper{T}$, which corresponds to a self-avoiding walk $w=(v,e_1, v_1\ldots,e_{\ell},v_\ell)$, we partition all self-avoiding walks $w'=(v,e_1, v_1\ldots,e_{\ell},v_\ell,e',v')$ in $\hyper{H}$ which extends $w$, into sets according to which hyperedge they use to extend the original walk $w$, so that self-avoiding walks within the same sets extends $w$ with the same hyperedge $e'$. For every set, we create a distinct hyperedge in $\hyper{T}$ incident to $u$ which contains the children of $u$ corresponding to the self-avoiding walks within that set.

\item
We further modify the hypertree $\hyper{T}$ obtained from the above two steps according to how cycles are closed. For any vertex $u$ in $\hyper{T}$ corresponding to a self-avoiding walk $w=(v,e_1, v_1\ldots,e_{\ell},v_\ell)$ which can be extended to a cycle-closing walk $w'=(v,e_1, v_1\ldots,e_{\ell},v_\ell,e',v')$ such that $v'\in\{v,v_1,\ldots,v_{\ell-1}\}$, denoted by $e''$ the hyperedge in $w$ starting that cycle, if it holds that $e'<_{v'}e''$, i.e.~the hyperedge ending the cycle is ranked higher than the hyperedge starting the cycle by the cycle-closing vertex, then vertex $u$ along with all its descendants  in $\hyper{T}$ are deleted from $\hyper{T}$. Any hyperedges whose size becomes 1 because of this step are also deleted from $\hyper{T}$.

\end{enumerate}

The construction is illustrated in Figure~\ref{fig:saw}.

\begin{figure}
\centering
\includegraphics[scale=.75]{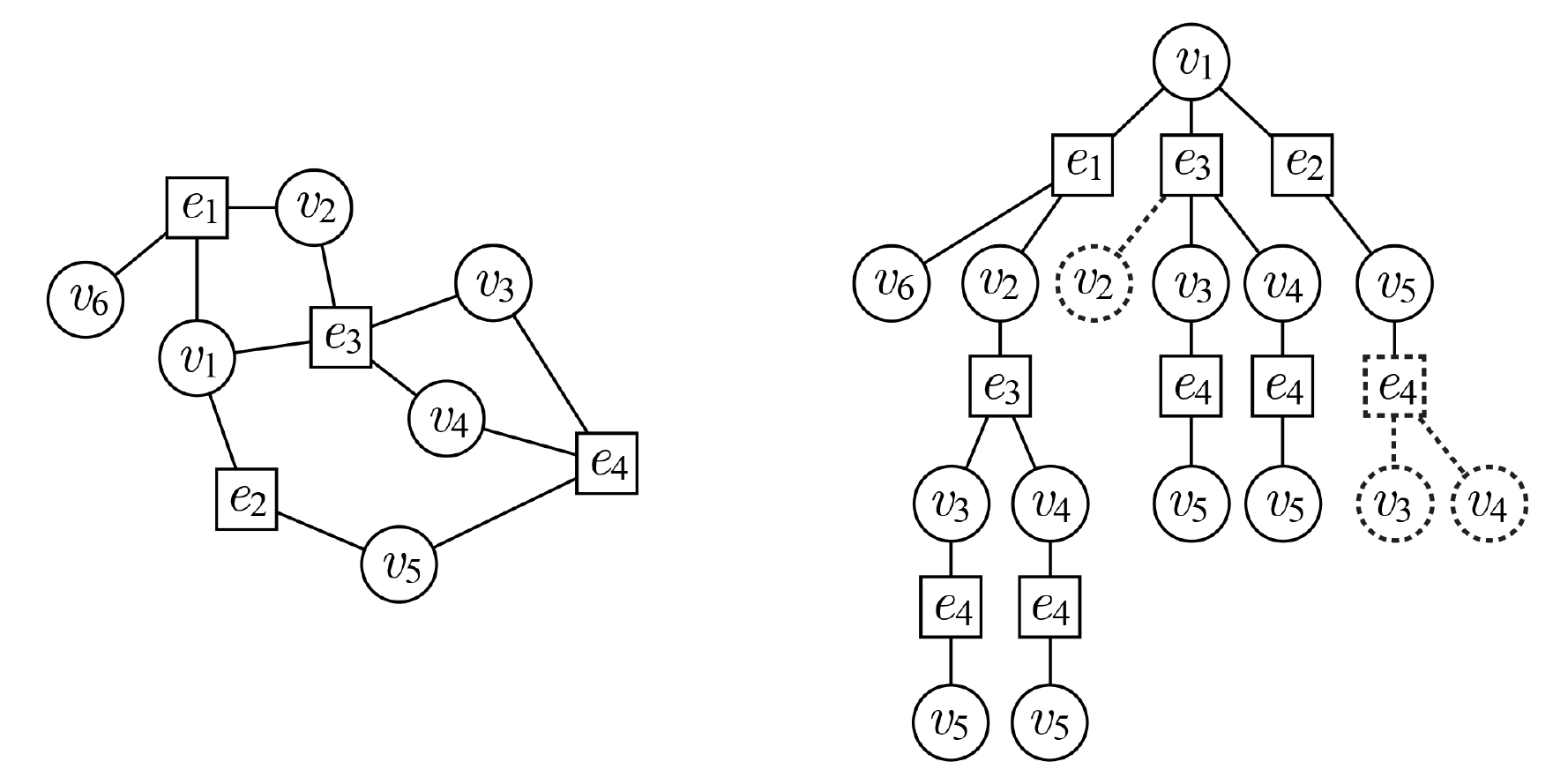}
\caption{\footnotesize{
The construction of $\TSAW$. On the left is a hypergraph $\hyper{H}$ and on the right is $\TSAW(\hyper{H}, v_1)$, both drawn as incident graphs. The ordering of the hyperedges incident to each vertex in $\hyper{H}$ is given by the subscripts. Each vertex or hyperedge in $\TSAW$ is labeled by the name of the vertex or hyperedge to which it is identified in $\hyper{H}$. 
Dashed vertices are the ones deleted according to the ordering of incident hyperedges at the cycle-closing vertices.
Dashed hyperedge is deleted because its size becomes $1$.
}}
\label{fig:saw}
\end{figure}

We consider the Gibbs measure of a rooted hypertree $\hyper{T}$ with activity $\lambda$, and use $\ptree{\hyper{T}}{\sigma_{\Lambda}}$ to denote the marginal probability of the root of $\hyper{T}$ being occupied conditioning on $\sigma_\Lambda$.

Note that each vertex $u$ in $\TSAW(\hyper{H}, v)$ can be naturally identified (many-to-one) to the vertex in $\hyper{H}=(V,E)$ at which the self-avoiding walk corresponding to $u$ ends, thus a configuration $\sigma_\Lambda$ partially specified on a subset $\Lambda\subset V$ of vertices in $\hyper{H}$ can be directly translated to a partially specified configuration in $\TSAW(\hyper{H}, v)$ through the one-to-many association. We abuse the notation and still denote the resulting configuration in $\hyper{T}=\TSAW(\hyper{H}, v)$ as $\sigma_\Lambda$, thus $\ptree{\hyper{T}}{\sigma_{\Lambda}}$ is well-defined.

\begin{theorem}\label{thm-peq}
Let $\hyper{H}=(V,E)$ be a hypergraph and $\lambda>0$. For any $v\in V$, $\Lambda \subseteq V$ and $\sigma_{\Lambda}\in\{0,1\}^\Lambda$, it holds that $p_{\hyper{H},v}^{\sigma _{\Lambda}} = \ptree{\hyper{T}}{\sigma_{\Lambda}}$ where $\hyper{T} = \TSAW(\hyper{H}, v)$.
\end{theorem}

\begin{proof}
The proof follows the same routine as that of Weitz~\cite{weitz2006counting}, with some extra cares to be taken to avoid the complications caused by hypergraphs.

Denote $R_{\hyper{H},v}^{\sigma_\Lambda}(\lambda) = p_{\hyper{H},v}^{\sigma_{\Lambda}} / (1 - p_{\hyper{H},v}^{\sigma_{\Lambda}})$ for the ratio between the probability that $v$ in $\hyper{H}$ is occupied and unoccupied conditioning on configuration $\sigma_\Lambda$ of $\Lambda \subset V$.
We write $R_\hyper{T}^{\sigma_\Lambda} = R_{\hyper{T},v}^{\sigma_\Lambda}$ when $v$ is unambiguously the root of $\hyper{T}$.

Let $d$ be the degree of the root of $\hyper{T}$. Suppose that there are $k_i$ children contained in $i$-th child-edge, where the order is determined during the construction of $\TSAW(\hyper{H},v)$. $\hyper{T}_{ij}$ is the subtree rooted at the $j$-th child in the $i$-th child-edge. Let $\Lambda_{ij} = \Lambda \cap \hyper{T}_{ij}$ and $\sigma_{\Lambda_{ij}}$ be the restriction of $\sigma_\Lambda$ on $\Lambda_{ij}$.
Applying the tree recursion \eqref{eq:tree-recursion} for the self-avoiding walk tree $\hyper{T}$, we have
\begin{align}
\label{eq:tree-recur}
R_\hyper{T}^{\sigma_\Lambda} = \lambda \prod _{i=1}^d \frac{1}{1 + \sum_{i=1}^{k_i} R_{\hyper{T}_{ij}}^{\sigma_{\Lambda_{ij}}}},
\end{align}
This defines a recursive procedure for calculating $R_\hyper{T}^{\sigma_\Lambda}$. The base cases are naturally defined when $v$ lies in $\Lambda$, in which case $R_\hyper{T}^{\sigma_\Lambda} = 0$ if $v$ is fixed unoccupied or $R_\hyper{T}^{\sigma_\Lambda} = \infty$ if it is fixed occupied, or when $v$ has no child, in which case $R_\hyper{T}^{\sigma_\Lambda} = \lambda$.

In the following we describe our procedure for calculating $R_{\hyper{H},v}^{\sigma_\Lambda}$ at $v$ in the original hypergraph $\hyper{H}$.
The problem comes that the ratio at different neighbors of $v$ may still depend on each other when we fix the value at $v$ since there may exist cycles in $\hyper{H}$.
We resolve this problem by editing the original hypergraph around $v$ and imposing appropriate conditions for each neighbor of $v$.

Let $\hyper{H}^v$ be the same hypergraph as $\hyper{H}$ except that vertex $v \in V$ is substituted by $d$ vertices $v_1, v_2, ..., v_d$, where $d$ is the degree of $v$. Each vertex $v_i$ is contained into a single hyperedge $e_i$, where $e_i$ is the $i$-th hyperedge connecting $v$, and the order here is the same as the one determined in the definition of $\TSAW(\hyper{H}, v)$. At the same time, we associated each $v_i$ with an activity of $\lambda^{1/d}$ rather than $\lambda$. It is now clear to see that an independent set in $\hyper{H}$ with $v$ occupied has the same weight as the corresponding independent set in $\hyper{H}^v$ with all the $v_i$ occupied, and so is the case when $v$ is unoccupied. Therefore, $R_{\hyper{H},v}^{\sigma_\Lambda}$ equals to the ratio between the probabilities in $\hyper{H}^v$ with all $v_i$ ($1 \le i \le d$) being occupied and unoccupied, conditioning on $\sigma_\Lambda$. Let $\tau_i$ be the configuration for vertex $v_i$ in which the values of $v_j$ are fixed to occupied if $j < i$ and unoccupied if $j > i$. We can then write this in a form of telescopic product:
\[
R_{\hyper{H},v}^{\sigma_\Lambda} = \prod_{i=1}^{d} R_{{\hyper{H}^v}, v_i}^{{\sigma_\Lambda} {\tau_i}},
\]
where ${\sigma_\Lambda} {\tau_i}$ means the combination of the two configurations ${\sigma_\Lambda}$ and ${\tau_i}$.

We can obtain the value of $R_{{\hyper{H}^v}, v_i}^{{\sigma_\Lambda} {\tau_i}}$ by further fix vertices in $e_i$, the hyperedge containing $v_i$. Since now $v_i$ is contained only in $e_i$, we can see that
\[
R_{{\hyper{H}^v}, v_i}^{{\sigma_\Lambda} {\tau_i}} =
\frac{ \lambda^{1/d} }{ 1 + \sum_{j=1}^{k_i} R_{{\hyper{H}^v \slash v_i}, u_{ij}}^{ {\sigma_\Lambda} {\tau_i} {\rho_{ij}} } },
\]
where $k_i$ is the number of the vertices other than $v_i$ which is incident to $e_i$ and $\rho_{ij}$ is the configuration at vertices of $e_i$ in which all the vertices $u_{ij'}$ other than $u_{ij}$ are fixed to unoccupied.

Combining above two equations, we get a recursive procedure for calculating $R_{\hyper{H},v}^{\sigma_\Lambda}$ in the same manner that equation \eqref{eq:tree-recur} has:
\begin{align}
\label{eq:graph-recur}
R_{\hyper{H},v}^{\sigma_\Lambda} = \lambda \prod_{i=1}^{d}
\frac{ 1 }{ 1 + \sum_{j=1}^{k_i} R_{{\hyper{H}^v \slash v_i}, u_{ij}}^{ {\sigma_\Lambda} {\tau_i} {\rho_{ij}} } }.
\end{align}
Notice that the recursion does terminate, since the number of unfixed vertices reduces at least by one in each step because in calculating $R_{{\hyper{H}^v \slash v_i}, u_{ij}}^{ {\sigma_\Lambda} {\tau_i} {\rho_{ij}} }$ all copies $v_{i'}$ of $v$ is either fixed (when $i' \neq i$) or erased (when $i' = i$) from the hypergraph ${\hyper{H}^v \slash v_i}$.

We now show that the procedure described above for calculating $R_{\hyper{H},v}^{\sigma_\Lambda}$ results in the same value as using the hypertree procedure for $\TSAW(\hyper{H}, v)$ with corresponding condition of $\sigma_\Lambda$ imposed on it. First notice that the calculation carried out by the two procedure is the same, since they share the same function (Equation \eqref{eq:tree-recur} and \eqref{eq:graph-recur}) when we view them as recursive calls.
Furthermore, we have the same stopping values for the both recursive procedures.
During constructing $\TSAW(\hyper{H},v)$, if node $u$ corresponding to walk 
is not included in the hypertree, which is equivalent to fix $u$ to unoccupied in the sense of causing the same effect on the ratio of occupation to its parent node. And when node $u$ in the hypertree corresponding to a self-avoiding walk $w=(v,e_1,v_1\ldots,e_{\ell},v_\ell)$, with that $w$ can be extended as $w' = (w,e_{\ell+1},v_{\ell+1})$ to a cycle-closing vertex $v_{\ell+1}=v_i$ for some $0\le i<\ell$ via a new hyperedge $e_{\ell+1}\not\in\{e_0,e_1,\ldots,e_{\ell}\}$, and $e_{\ell+1}<_{v_{i}}e_i$, then the node $u$ along with all its descendants are deleted. This gives the equivalent effect to parent node of $u$ as if $u$ is fixed to unoccupied, or one of the children of $u$ (i.e. the node corresponding to $w'$) to occupied, which is what we did to fix the vertices $v_j$ for $j<i$ in $\tau_i$. 
Eliminating a hyperedge with no child also does not affect the final value of $R_{\hyper{T}}^{\sigma_\Lambda}$.

Thus, what is left to complete the proof is to show that the hypertree $\TSAW(\hyper{H}^v\slash v_{i}, u_{ij})$ with $(\sigma_\Lambda \tau_i \rho_{ij})$'s corresponding condition imposed on it is exactly the same as the subtree of $\TSAW(\hyper{H}, v)$ rooted at the $j$-th child vertex of the $i$-th child-edge of the root with $\sigma_\Lambda$'s corresponding condition imposed on it. This is enough because then the resulting values are the same for both procedures by induction. The observation is that both trees are the hypertree of all self-avoiding walks in $\hyper{H}$ starting at $u_{ij}$, except that $\TSAW(\hyper{H}^v\slash v_{i}, u_{ij})$ has some extra vertices which are fixed to be occupied or unoccupied depending on whether the corresponding walk reaches $v$ via a higher or lower ranked hyperedge, or reaches $i$-th hyperedge of $v$, which results in the same probability of occupation at the root.
\end{proof}

A hypergraph $\hyper{H}$ is a sub-hypergraph of another hypergraph $\hyper{G}$ if the incidence graph of $\hyper{H}$ is a subgraph of that of $\hyper{G}$, and for hypertrees this is samely defined. Note that for hypergraphs, a subgraph is not necessarily formed by a sub-collection of hyperedges, but maybe also by sub-hyperedges.
The $\TSAW$ of a hypergraph $\hyper{H}$ with maximum edge-size at most $k+1$ and maximum degree at most $d+1$ is sub-hypertree of $\RHtree{k}{d}$.

\begin{proposition}\label{lem-pru}
Let $\hyper{T}_0=(V_0, E_0)$ be a rooted hypertree and $\hyper{T}=(V,E)$ its sub-hypertree with the same root. 
For any $\Lambda \subseteq V$ and any $\sigma_{\Lambda}\in\{0,1\}^\Lambda$, there exists a configuration $\sigma_{\Lambda_0}\in\{0,1\}^{\Lambda_0}$ for $\Lambda\subseteq\Lambda_0 \subseteq V_0$, extending the configuration $\sigma_\Lambda$, such that $\ptree{\hyper{T}}{\sigma_{\Lambda}} = \ptree{\hyper{T}_0}{\sigma_{\Lambda_0}}$. 
\end{proposition}

The configuration $\sigma_{\Lambda_0}$ just extends $\sigma_\Lambda$ by fixing all the vertices missing in $\hyper{T}$ (actually only those who are  closest to the root along each path) to be unoccupied.

Theorem \ref{thm-saw} follows immediately from Theorem~\ref{thm-peq} and Proposition~\ref{lem-pru}.

\begin{proof}[Proof of Theorem~\ref{thm-saw}]
Given any hypergraph $\hyper{H}$ of maximum edge-size at most $(k+1)$ and maximum degree at most $(d+1)$, by Theorem~\ref{thm-peq} we have $| \pgraphhv{\sigma_\Lambda} - \pgraphhv{\tau_\Lambda} | = | \ptreehv{\sigma_\Lambda} - \ptreehv{\tau_\Lambda} | $ where $\hyper{T}=\TSAW(\hyper{H},v)$.
The distance from the root $v$ to any vertex $u$ in $\hyper{T}$ is no shorter than the distance $\hyper{H}$ between $v$ and the vertex in $\hyper{H}$ to which $u$ is identified. So the SSM with rate $\delta(\cdot)$ on $\hyper{T}$ implies that on the hypergraph $\hyper{H}$.

Since $\hyper{H}$ has maximum edge-size at most $k+1$ and maximum degree at most $d+1$, its SAW-tree $\hyper{T}=\TSAW(\hyper{H},v)$ is a sub-hypertree of $\RHtree{k}{d}$. Thus by Proposition~\ref{lem-pru},  we have $| \ptreehv{\sigma_\Lambda} - \ptreehv{\tau_\Lambda} |=| \mathbb{P}_{\RHtree{k}{d}}^{\sigma_{\Lambda_0}} - \mathbb{P}_{\RHtree{k}{d}}^{\tau_{\Lambda_0}} |$ for some $\sigma_{\Lambda_0}, \tau_{\Lambda_0}$ extending $\sigma_{\Lambda}, \tau_{\Lambda}$. The SSM on $\RHtree{k}{d}$ with rate $\delta(\cdot)$ implies that on $\hyper{T}$, which implies the same on the original hypergraph $\hyper{H}$.
\end{proof}

\section{Strong spatial mixing}
\label{sec:SSM}

In this section, we show that for independent sets of the infinite $(k+1)$-uniform  $(d+1)$-regular hypertree $\RHtree{k}{d}$, weak spatial mixing implies strong spatial mixing at almost the same rate.

\begin{theorem}\label{theorem-srj-ssm-IS}
For every positive integers $d,k$ and any $\lambda$, if the independent sets of the infinite $(k+1)$-uniform $(d+1)$-regular hypertree $\RHtree{k}{d}$ with activity $\lambda$ exhibits weak spatial mixing with rate $\delta(\cdot)$ then it also exhibits strong spatial mixing with rate $\frac{(1+\lambda)\left(\lambda+(1+k\lambda)^{d+1}\right)}{\lambda}\delta(\cdot)$.
\end{theorem}
By Theorem~\ref{thm-saw}, this implies the strong spatial mixing with the same rate on all hypergraphs of maximum degree at most $d+1$ and maximum size of hyperedges at most $k+1$. 

Unlike most known strong spatial mixing results, where the spatial mixing is usually established by an analytic approach with help of potential functions, our proof of Theorem~\ref{theorem-srj-ssm-IS} adopts the combinatorial argument used in Weitz's original proof of SSM for the hardcore model~\cite{weitz2006counting}. Weitz's approach gives us a stronger result: It explicitly gives the extremal case for WSM as well as SSM among a family of hypergraphs with bounded maximum degree and bounded maximum edge-size. It can also easily give us the SSM behavior when at the critical threshold.

Assume the hypertree $T=\RHtree{k}{d}$ is rooted at some vertex $v$. 
For $\ell>0$, let $R^+_\ell$ and $R^-_\ell$ denote the respective maximum and minimum values of $R_{T}^{\sigma}$ achieved by a boundary condition $\sigma$ that fixes the states of all vertices at level $\ell$. By the monotonicity of the tree recursion, it is easy to see that $R^+_\ell$ (or $R^-_\ell$) is computed by the tree recursion with initial values at all vertices at level $\ell$ to be $\infty$ (or 0) if $\ell$ is even, and $0$ (or $\infty$) if $\ell$ is odd, with the root $v$ being at level 0.
\footnote{Note that although the all-$\infty$ initial values corresponds to a boundary condition $\sigma$ that fixes all vertices at level $\ell$ to be occupied, which may no longer be a valid independent set in the hypertree, the $R^{\pm}_\ell$ achieved by this choice of initial values is actually the same as the $R_{T}^{\sigma}$ with a boundary condition $\sigma$ that fixes exactly one vertex per hyperedge to be occupied at level $\ell$.}

It is easy to see that fixing a vertex $u$ in $T$ to be occupied has the same effect as fixing $u$'s parent to be unoccupied, therefore to prove SSM, it is sufficient to prove the decay of correlation conditioning on a subset of vertices in $T$ fixed to be unoccupied.
Another key observation from the tree recursion is that fixing a vertex $u$ in $T$ to be unoccupied has the same effect as having a local activity $\lambda_u=0$ at vertex $u$. 
Now consider a vector $\vec{\lambda}$ that assigns every vertex $u$ in $T=\RHtree{k}{d}$ a local activity $\lambda_u$. 
Let $R^{+}_\ell(\vec{\lambda})$ and $R^{-}_\ell(\vec{\lambda})$ be accordingly defined as the respective extremal values of $R^\sigma_{T}(\vec{\lambda})$ achieved by boundary conditions $\sigma$ fixing all vertices at level $\ell$ in the tree $T=\RHtree{k}{d}$ equipped with the nonuniform activities $\vec{\lambda}$.
Clearly, by the same monotonicity, $R^{\pm}_\ell(\vec{\lambda})$ can be computed from the tree recursion with a nonuniform activities $\vec{\lambda}$ with the same settings of initial values as the uniform case $R^{\pm}_\ell=R^{\pm}_\ell({\lambda})$.

The following theorem shows that basically the decay of correlation is dominated by the uniform activity case.
\begin{theorem}\label{theorem-srj-sens-de}
        Fix an arbitrary $\lambda\geq0$. Let $\vec{\lambda}$ be an assignment of activities to vertices of $\RHtree{k}{d}$ such that $0\leq \lambda_v \leq\lambda$ for every $v\in \RHtree{k}{d}$. For every $\ell \geq 1$ we have
        $$\frac{R^{+}_{\ell}(\vec{\lambda})}{R^{-}_{\ell}(\vec{\lambda})}\leq\frac{R^{+}_{\ell}}{R^{-}_{\ell}}$$
\end{theorem}
Translated to the language of subtrees, the theorem means that the extremal case of WSM among a family hypertrees with bounded maximum degree and maximum edge-size, is given by the uniform regular tree with the highest degree and edge-size in the family.  
Technically, Theorem~\ref{theorem-srj-sens-de} measures the decay of correlation in terms of $\log R=\log\frac{p}{1-p}$. Note that for $\ell\ge 2$, it always holds that $p^+_\ell(\lambda) \leq \frac{\lambda}{1+\lambda}$ and $p^-_\ell(\lambda) \geq \frac{\lambda}{\lambda+(1+k\lambda)^{d+1} }$, where $R_\ell^\pm=\frac{p_\ell^\pm}{1-p_\ell^\pm}$. Theorem~\ref{theorem-srj-sens-de} implies Theorem~\ref{theorem-srj-ssm-IS}.

We now consider a slightly different hypertree which is exactly the same as $\RHtree{k}{d}$ except that the degree of root is $d$. Denote this hypertree as $\VHtree{k}{d}$. 

\begin{lemma}\label{lemma-srj-hyp}
        For every integer $\ell\geq1$ and any assignment of activities $\vec{\lambda}$ to vertices of $\VHtree{k}{d}$ such that $0\le\lambda_v\le\lambda$ for every vertex $v$, the following two inequalities hold:
        \begin{align}
            \frac{R^{+}_{\ell}(\vec{\lambda})}{R^{-}_{\ell}(\vec{\lambda})}
            &\leq\frac{R^{+}_{\ell}}{R^{-}_{\ell}},\label{eq:srj-ssm-h1}\\
            \frac{1+kR^{+}_{\ell}(\vec{\lambda})}{1+kR^{-}_{\ell}(\vec{\lambda})}
            &\leq\frac{1+kR^{+}_{\ell}}{1+kR^{-}_{\ell}},\label{eq:srj-ssm-h2}
        \end{align}
        with the convention $0/0=1$ and $\infty=\infty$.
\end{lemma}

\begin{proof}[Proof of Lemma~\ref{lemma-srj-hyp}]
The proof is by an induction on $\ell$. The proof is similar to that of Weitz~\cite{weitz2006counting} except for the parts dealing with hyperedges.

First consider the exceptional cases when the denominators in \eqref{eq:srj-ssm-h1} may be zero.
Assume $R^{+}_{\ell}(\vec{\lambda}) = R^{-}_{\ell}(\vec{\lambda})=0$, which only happens when the activity of the root is zero. We adopt the convention that $\frac{R^{+}_{\ell}(\vec{\lambda})}{R^{-}_{\ell}(\vec{\lambda})} = 1$. 
Assume $R^{-}_{\ell} = 0$, which only occurs when $\ell=1$. Then $R^{-}_{\ell}(\vec{\lambda})=0$ also holds, and by convention we have $\frac{R^{+}_{\ell}(\vec{\lambda})}{R^{-}_{\ell}(\vec{\lambda})}=\frac{R^{+}_{\ell}}{R^{-}_{\ell}}=\infty$. 
Note that these conventions are consistent with our induction, such that assuming the induction hypothesis $\frac{R^{+}_{\ell}(\vec{\lambda})}{R^{-}_{\ell}(\vec{\lambda})}\leq\frac{R^{+}_{\ell}(\lambda)}{R^{-}_{\ell}(\lambda)}$, for any $k$ assignments of activities $0\le \vec{\lambda_{1}}, \vec{\lambda_{2}},...,\vec{\lambda_{k}}\le \lambda$, there exists $\alpha\geq0$ such that $\sum_{i=1}^{k}R^{-}_{\ell}(\vec{\lambda_{i}})=\alpha kR^{-}_{\ell}$ and $\sum_{i=1}^{k}R^{+}_{\ell}(\vec{\lambda_{i}})\leq\alpha kR^{+}_{\ell}$.

For the basis, $\ell=1$. We have $R^{-}_{\ell}(\vec{\lambda})\ge R^{-}_{\ell}=0$, $R^{+}_{\ell}=\lambda$, and $R^{+}_{\ell}(\vec{\lambda})=\lambda_r$ where $\lambda_r$ is the activity of the root. The hypotheses~\eqref{eq:srj-ssm-h1} and~\eqref{eq:srj-ssm-h2} are true since $\lambda_r\leq \lambda$.

Assume~\eqref{eq:srj-ssm-h1} and~\eqref{eq:srj-ssm-h2} are true for an $\ell\ge 1$. We will show that they are true for $\ell+1$. The following recursion holds
\begin{align*}
\frac{R^{+}_{\ell+1}(\vec{\lambda})}{R^{-}_{\ell+1}(\vec{\lambda})} = \prod_{i=1}^{d}\frac{1+\sum_{j=1}^{k}R^{+}_{\ell}(\vec{\lambda}_{ij})}{1+\sum_{j=1}^{k}R^{-}_{\ell}(\vec{\lambda}_{ij})}
    =\prod_{i=1}^{d}\frac{\sum_{j=1}^{k}(1+kR^{+}_{\ell}(\vec{\lambda}_{ij}))}{\sum_{j=1}^{k}(1+kR^{-}_{\ell}(\vec{\lambda}_{ij}))},
\end{align*}
where $\vec{\lambda}_{ij}$ stands for the restriction of the assignment $\vec{\lambda}$ to the subtree of $\RHtree{k}{d}$ rooted at the $j$-th child in the $i$-th edge incident to the root.
By induction hypothesis~\eqref{eq:srj-ssm-h2}, we have $\frac{1+kR^{+}_{\ell}(\vec{\lambda}_{ij})}{1+kR^{-}_{\ell}(\vec{\lambda}_{ij})}\leq\frac{1+kR^{+}_{\ell}}{1+kR^{-}_{\ell}}$, so immediately, 
\begin{align*}
   \frac{R^{+}_{\ell+1}(\vec{\lambda})}{R^{-}_{\ell+1}(\vec{\lambda})} = \prod_{i=1}^{d}\frac{\sum_{j=1}^{k}(1+kR^{+}_{\ell}(\vec{\lambda}_{ij}))}{\sum_{j=1}^{k}(1+kR^{-}_{\ell}(\vec{\lambda}_{ij}))}
    \leq\left(\frac{1+kR^{+}_{\ell}}{1+kR^{-}_{\ell}}\right)^{d}=\frac{R^{+}_{\ell+1}}{R^{-}_{\ell+1}},
\end{align*}
where the inequality is due to the simple fact that if $a_i\ge b_i> 0$ and $\frac{a_i}{b_i}\leq t$ for all $i$, then $\frac{\sum^{n}_{i=1}a_i}{\sum^{n}_{i=1}b_i}\leq t$.

This proves~\eqref{eq:srj-ssm-h1} for $\ell+1$. Next we will prove~\eqref{eq:srj-ssm-h2}. Recall the tree recursions:
\begin{align*}
R^{\pm}_{\ell+1}(\vec{\lambda})=\lambda_r\prod\limits^{d}_{i=1}\frac{1}{1+\sum\limits^{k}_{j=1}R^{\mp}_{\ell}(\vec{\lambda}_{ij})}
\quad\text{ and }\quad
R^{\pm}_{\ell+1}=\lambda\prod\limits^{d}_{i=1}\frac{1}{1+kR^{\mp}_{\ell}},
\end{align*}
where $\lambda_r$ is the local activity assigned by $\vec{\lambda}$ to the root $r$.
Observe that if $\sum_{j=1}^{k}R^-_\ell(\vec{\lambda}_{ij}) \geq kR^-_\ell$ for all $i \in [d]$, then $R^+_{\ell+1}(\vec{\lambda}) \leq R^+_{\ell+1}{}$, which combined with~\eqref{eq:srj-ssm-h1} for $\ell+1$ that we just proved above, would give us that $\frac{1+kR^{+}_{\ell+1}(\vec{\lambda})}{1+kR^{-}_{\ell+1}(\vec{\lambda})}\leq\frac{1+kR^{+}_{\ell+1}}{1+kR^{-}_{\ell+1}}$. In this good case, the hypothesis~\eqref{eq:srj-ssm-h2} easily holds for $\ell+1$. We then show that the opposite case where $\sum^{k}_{j=1}R^{-}_{\ell}(\vec{\lambda}_{ij}) \leq kR^{-}_{\ell}$ for all $i$ represents the worst possible case, and it is enough to prove the hypothesis~\eqref{eq:srj-ssm-h2} under this condition.
To see this, assume to the contrary that for some $i_0$, $\sum^{k}_{j=1}R^{-}_{\ell}(\vec{\lambda}_{i_{0}j}) > kR^{-}_{\ell}$. 
We then construct a $\vec{\lambda}'$ that satisfies $\sum^{k}_{j=1}R^{-}_{\ell}(\vec{\lambda}_{i_{0}j}') \le kR^{-}_{\ell}$ and has an even worse ratio between $1+kR^+_{\ell+1}$ and $1+kR^-_{\ell+1}$.
Let $\vec{\lambda}'$ be the same as $\vec{\lambda}$ except that for every $j \in [k]$, $\vec{\lambda}_{i_{0}j}'$ is uniform and is equal to $\lambda$ everywhere.  
Clearly, it holds that $\sum^{k}_{j=1}R^{-}_{\ell}(\vec{\lambda}_{i_{0}j}') = kR^{-}_{\ell}$.
On the other hand, by the induction hypothesis, for every $j$ we have $\frac{1+kR^{+}_{\ell}(\vec{\lambda}_{i_{0}j})}{1+kR^{-}_{\ell}(\vec{\lambda}_{i_{0}j})} \leq \frac{1+kR^{+}_{\ell}}{1+kR^{-}_{\ell}}$, and hence
\[
    \frac{1+\sum_{j=1}^kR^{+}_{\ell}(\vec{\lambda}_{i_{0}j})}{1+\sum_{j=1}^kR^{-}_{\ell}(\vec{\lambda}_{i_{0}j})}=\frac{\sum^{k}_{j=1}(1+kR^{+}_{\ell}(\vec{\lambda}_{i_{0}j}))}{\sum^{k}_{j=1}(1+kR^{-}_{\ell}(\vec{\lambda}_{i_{0}j}))} \leq \frac{1+kR^{+}_{\ell}}{1+kR^{-}_{\ell}},
\]
where again the inequality uses the fact that if $a_i\ge b_i> 0$ and $\frac{a_i}{b_i}\leq t$ for all $i$, then $\frac{\sum^{n}_{i=1}a_i}{\sum^{n}_{i=1}b_i}\leq t$.

    Note that $\vec{\lambda}'$ only changes the activities of all the subtrees rooted by the the children in the $i_0$-th edge of the root. So we have
\begin{align*}
        \frac{R^+_{\ell+1}(\vec{\lambda})}{R^-_{\ell+1}(\vec{\lambda})} =& \prod^{d}_{i=1}\frac{1+\sum^{k}_{j=1}R^+_{\ell}(\vec{\lambda}_{ij})}{1+\sum^{k}_{j=1}R^-_{\ell}(\vec{\lambda}_{ij})} \leq \frac{R^+_{\ell+1}(\vec{\lambda}')}{R^{-}_{\ell+1}(\vec{\lambda}')},\\
\text{ and }\quad       
R^{+}_{\ell+1}(\vec{\lambda})=&\lambda_r\prod\limits^{d}_{i=1}\frac{1}{1+\sum\limits^{k}_{j=1}R^{-}_{l}(\vec{\lambda}_{ij})} \leq R^{+}_{l+1}(\vec{\lambda}').
\end{align*}
Combine the two inequalities, we have $\frac{1+kR^{+}_{\ell+1}(\vec{\lambda})}{1+kR^{-}_{\ell+1}(\vec{\lambda})} \leq \frac{1+kR^{+}_{\ell+1}(\vec{\lambda}{'})}{1+kR^{-}_{\ell+1}(\vec{\lambda}{'})}$, an even worse case. 
So for the rest we only need to consider the case in which for every $i$, $\sum^{k}_{j=1}R^{-}_{\ell}(\vec{\lambda}_{ij}) \leq kR^{-}_{\ell}$.

For every $1 \leq i \leq d$,  we can choose $0\leq \alpha_i \leq 1$ so that $\sum^{k}_{j=1}R^{-}_{\ell}(\vec{\lambda}_{ij}) = \alpha_{i}kR^{-}_{\ell}$. Fix $i$ and by the induction hypothesis,   for every $1 \leq j \leq k$ we have $\frac{R^{+}_{\ell}(\vec{\lambda}_{ij})}{R^{-}_{\ell}(\vec{\lambda}_{ij})} \leq \frac{R^{+}_{\ell}}{R^{-}_{\ell}}$. If all $R^{-}_{\ell}(\vec{\lambda}_{ij})$ equal zero, then $\sum^{k}_{j=1}R^{+}_{\ell}(\vec{\lambda}_{ij}) \leq \alpha_{i}kR^{+}_{\ell}$ trivially holds as we argued in the beginning. Otherwise, note that since not all $R^{-}_{\ell}(\vec{\lambda}_{ij})$ are zero, we must have $\ell>1$, so if $R^{-}_{\ell}(\vec{\lambda}_{ij})=0$ then $R^{+}_{\ell}(\vec{\lambda}_{ij})=0$. Thus we also have $\frac{\sum^{k}_{j=1}R^{+}_{\ell}(\vec{\lambda}_{ij})}{\sum^{k}_{j=1}R^{-}_{\ell}(\vec{\lambda}_{ij})} \leq \frac{R^{+}_{\ell}}{R^{-}_{\ell}}$. In conclusion, in both cases we have $\sum^{k}_{j=1}R^{+}_{\ell}(\vec{\lambda}_{ij}) \leq \alpha_{i}kR^{+}_{\ell}$. 

Observe that $\lambda_r \le \lambda$ and $\prod^{d}_{i=1}\frac{1+\sum^{k}_{j=1}R^{+}_{\ell}(\vec{\lambda}_{ij})}{1+\sum^{k}_{j=1}R^{-}_{\ell}(\vec{\lambda}_{ij})} \ge 1$, it holds that
 \begin{align*}
    \frac{1+kR^{+}_{\ell+1}(\vec{\lambda})}{1+kR^{-}_{\ell+1}(\vec{\lambda})}=
    \frac{1+k\lambda_r\prod^{d}_{i=1}\frac{1}{1+\sum^{k}_{j=1}R^{-}_{\ell}(\vec{\lambda}_{ij})}}
    {1+k\lambda_r\prod^{d}_{i=1}\frac{1}{1+\sum^{k}_{j=1}R^{+}_{\ell}(\vec{\lambda}_{ij})}}
    \leq \frac{1+k\lambda\prod^{d}_{i=1}\frac{1}{1+\alpha_{i}kR^{-}_{l}}}
    {1+k\lambda\prod^{d}_{i=1}\frac{1}{1+\alpha_{i}kR^{+}_{\ell}}}.
   \end{align*}
Now it is enough to show that for every $\vec{\alpha}$ such that $0 \leq \alpha_i \leq 1$ for all $1 \leq i \leq d$, it holds that
    \begin{align*}
    \frac{1+k\lambda\prod^{d}_{i=1}\frac{1}{1+\alpha_{i}kR^{-}_{l}}}
    {1+k\lambda\prod^{d}_{i=1}\frac{1}{1+\alpha_{i}kR^{+}_{\ell}}}
    \leq\frac{1+kR^{+}_{\ell+1}}{1+kR^{-}_{\ell+1}},
    \end{align*}
which is equivalent to the following:    
\begin{align}
        1+k\lambda\prod^{d}_{i=1}\frac{1}{1+\alpha_ikR^-_{\ell}}-\frac{1+kR^+_{\ell+1}}{1+kR^-_{\ell+1}}-k\lambda
        \frac{1+kR^+_{\ell+1}}{1+kR^-_{\ell+1}}\prod^{d}_{i=1}\frac{1}{1+\alpha_ikR^+_{\ell}} \leq 0. \label{eq:srj-ssm-3}
    \end{align}
If $\alpha_i=1$ for every $i$ then the inequality~\eqref{eq:srj-ssm-3} trivially holds. By symmetry, we only need to show the LHS of~\eqref{eq:srj-ssm-3} is increasing in $\alpha_1$.
In fact, the partial derivative with respect to $\alpha_1$ of LHS in~\eqref{eq:srj-ssm-3} is:
\begin{align*}
    -\frac{k^{2}\lambda R^-_{\ell}}{1+\alpha_1kR^-_{\ell}}\prod^{d}_{i=1}\frac{1}{1+\alpha_ikR^-_\ell}+
    \frac{k^{2}\lambda\left(1+kR^+_{\ell+1}\right)R^+_{\ell}}{\left(1+kR^-_{\ell+1}\right)\left(1+\alpha_1kR^+_{\ell}\right)}\prod^{d}_{i=1}\frac{1}{1+\alpha_ikR^+_\ell}.
\end{align*}
To prove it is nonnegative, it is equivalent to show that
\begin{align}\label{eq:srj-ssm-4}
    \frac{\left(1+kR^+_{\ell+1}\right)R^+_\ell}{\left(1+kR^-_{\ell+1}\right)R^-_{\ell}} \geq \frac{1+\alpha_1kR^+_\ell}{1+\alpha_1kR^-_\ell}\prod^{d}_{i=1}\frac{1+\alpha_ikR^+_{\ell}}{1+\alpha_ikR^-_{\ell}}.
\end{align}
To prove~\eqref{eq:srj-ssm-4}, we first observe that $R^-_{\ell}$ is increasing in $\ell$ and $R^+_{\ell}$ is decreasing in $\ell$, which is exactly the same to prove as the same property of the hardcore model.  
This gives us the so-called sandwich condition:
\[
R^-_{\ell} \leq R^-_{\ell+1} \leq R^+_{\ell+1} \leq R^+_{\ell},
\] 
therefore $\frac{R^+_{\ell}}{R^-_{\ell}}\ge\frac{R^+_{\ell+1}}{R^-_{\ell+1}}$.
We are now ready to prove~\eqref{eq:srj-ssm-4}:
\begin{align*}
    \frac{\left(1+kR^+_{\ell+1}\right)R^+_{\ell}}{\left(1+kR^-_{\ell+1}\right)R^-_{\ell}} \geq& \frac{\left(1+kR^+_{\ell}\right)R^+_{\ell+1}}{\left(1+kR^-_{\ell}\right)R^-_{\ell+1}} \\
    =& \frac{1+kR^+_{\ell}}{1+kR^-_{\ell}}\prod^{d}_{i=1}\frac{1+kR^+_{\ell}}{1+kR^-_{\ell}} \\
    \geq& \frac{1+\alpha_1kR^+_{\ell}}{1+\alpha_1kR^-_{\ell}}\prod^{d}_{i=1}\frac{1+\alpha_ikR^+_{\ell}}{1+\alpha_ikR^-_{\ell}}.
\end{align*}
The last inequality uses the fact that $R^+_{\ell}\ge R^-_{\ell}$ and $0 \le \alpha_i\le 1$.
So~\eqref{eq:srj-ssm-4} is proved, which finishes our proof of Lemma~\ref{lemma-srj-hyp}. 
\end{proof}

Observe that the subtree rooted at the child of the root of $\RHtree{k}{d}$ is isomorphic to $\VHtree{k}{d}$. 
While at the root of $\RHtree{k}{d}$, we have
    \begin{align*}
    \frac{{R}^{+}_{\ell}(\vec{\lambda})}{{R}^{-}_{\ell}(\vec{\lambda})}=\prod^{d+1}_{i=1}\frac{1+\sum^{k}_{j=1}R^+_{\ell-1}(\vec{\lambda}_{ij})}
    {1+\sum^{k}_{j=1}R^-_{\ell-1}(\vec{\lambda}_{ij})}
    \leq \left(\frac{1+kR^+_{\ell-1}}{1+kR^-_{\ell-1}}\right)^{d+1}=\frac{{R}^{+}_{\ell}(\RHtree{k}{d})}{{R}^{-}_{\ell}(\RHtree{k}{d})}.
    \end{align*}
Together with Lemma~\ref{lemma-srj-hyp},  this completes our proof of Theorem~\ref{theorem-srj-sens-de}.

\paragraph*{Calculation of the decay rate.}  

The WSM rate of our model on the infinite $(k+1)$-uniform $(d+1)$-regular hypertree $\RHtree{k}{d}$ is the same as the hardcore model on the infinite $(d+1)$-regular tree with activity $k\lambda$. The WSM rate on regular tree has been addressed implicitly in the literature~\cite{kelly1985stochastic,spitzer1975markov}. 
Here we provide an analysis for the decay rate for the completeness of the paper.

Let $f_{d,k}(x) \defeq \frac{k\lambda}{(1+x)^d}$ denote the symmetric version of the tree recursion on $\VHtree{k}{d}$ and substituting $x=kR$. Since $f_{d,k}(x)$ is decreasing in $x$, it follows that there is a unique positive fixed point $\hat{x}$ such that $\hat{x}=f_{d,k}(\hat{x})$. Let $f'_{d,k}(\hat{x})=-\frac{d\hat{x}}{1+\hat{x}}$ be the derivative of $f_{d,k}(x)$ evaluated at the fixed point $x=\hat{x}$.
The following proposition is well known for hardcore model (see e.g.~\cite{kelly1985stochastic,spitzer1975markov}).

\begin{proposition}\label{proposition-uniqueness}
$|f'_{d,k}(\hat{x})|=\frac{d\hat{x}}{1+\hat{x}} \leq 1$ if and only if $\lambda \leq \lambda_c$. And $|f'_{d,k}(\hat{x})|< 1$ if $\lambda<\lambda_c$.
\end{proposition}

We write $f(x) = f_{d,k}(x)$ if $k$ and $d$ are clear in the context. The main result of this part is the following theorem.

\begin{theorem}\label{theorem-wsm-IS}
For any positive integers $d,k$, assuming $\lambda\le\lambda_c$, the model on $\RHtree{k}{d}$ exhibits weak spatial mixing with rate $\delta(\ell)$ such that for all sufficiently large $\ell$:
\begin{itemize}
\item if $\lambda<\lambda_c$,  then $\delta(\ell)\le C_1|f'(\hat{x})|^{\ell- 4}$;
\item if $\lambda=\lambda_c$, then $\delta(\ell) \leq \frac{C_2}{\sqrt{\ell-\ell_0}}$;
\end{itemize}
where $C_1, C_2,\ell_0>0$ are finite constants depending only on $k,d$ and $\lambda$.
\end{theorem}

Theorem~\ref{thm-saw}, Theorem~\ref{theorem-srj-ssm-IS} and~\ref{theorem-wsm-IS} together prove the SSM part  of Theorem~\ref{theorem-main}.

Denote $g(x) = f\left(f(x)\right) = k\lambda\left(1+\frac{k\lambda}{(1+x)^d}\right)^{-d}$. It is easy to see that $\hat{x}=g(\hat{x})$.  

\begin{lemma}\label{lemma-rate}
If $\lambda \leq \lambda_c$ then for any $x > \hat{x}$ we have $g(x) - g(\hat{x}) \leq f'(\hat{x})^2(x-\hat{x})$.
\end{lemma}
\begin{proof}
By the mean value theorem, for any $x > \hat{x}$, there exists a $z \in [\hat{x}, x]$ such that 
\begin{align}\label{eq:g-mvt}
	g(x) - g(\hat{x}) = \alpha(z)(x-\hat{x}),
\end{align}
where $\alpha(z) = g'(z) 
	= \frac{d^2k\lambda g(z)}{(1+z)^{d+1}+(1+z)k\lambda}$.
We will bound the maximum value of $\alpha(z)$ when $\lambda \leq \lambda_c$. Consider the derivative of $\alpha(z)$,
\begin{align*}
	\alpha'(z) 
	&= A(z)\left[(d-1)k\lambda - (1+z)^d\right],
\end{align*}
where $A(z) = \frac{d^2(d+1)k\lambda g(z)}{\left[(1+z)^{d+1} + (1+z)k\lambda \right]^2}>0$. Let $z^*=((d-1)k\lambda)^{1/d}-1$ be the solution of $(d-1)k\lambda =(1+z)^d$. Note that $\left[(d-1)k\lambda - (1+z)^d\right]$ is decreasing in $z$, therefore $\alpha(z) \leq \alpha(z^*)$ for all $z > 0$.
Due to proposition~\ref{proposition-uniqueness}, if $\lambda \leq \lambda_c$ then $|f'(\hat{x})| =\frac{d\hat{x}}{1+\hat{x}} \leq 1$ and hence $\hat{x} \leq \frac{1}{d-1}$, thus $\alpha'(\hat{x})=A(\hat{x})[(d-1)k\lambda-\frac{k\lambda}{\hat{x}}]\le 0$,
which means $\hat{x} \geq z^*$ and $\alpha(z)$ is decreasing in $z$ for any $z\ge \hat{x}$.  
On the other hand, we have $\alpha(\hat{x})  = f'(\hat{x})^2$.
Thus for any $z \ge\hat{x}$, we have $\alpha(z) \le \alpha(\hat{x}) = f'(\hat{x})^2$. 
Plug it into~\eqref{eq:g-mvt}. The lemma is proved.
\end{proof}

\begin{proof}[Proof of Theorem~\ref{theorem-wsm-IS}]
It holds that $R^+_2=R^+_1 = \lambda > \hat{x}/k$. Note that $kR^+_\ell = g(kR^+_{\ell-2})$. Due to the monotonicity of $g(x)$, we have $\hat{x}<kR^+_\ell \le k\lambda$ for every $\ell \geq 1$. 

Consider the case $\lambda < \lambda_c$. First consider the $(k+1)$-uniform $d$-ary hypertree $\VHtree{k}{d}$.
By the mean value theorem and Lemma~\ref{lemma-rate} we have
\begin{align*}
	kR^+_\ell - \hat{x} =& g(kR^+_{\ell-2}) - g(\hat{x}) \leq f'(\hat{x})^2\left(kR^+_{\ell-2} - \hat{x} \right).
\end{align*}
We apply this inequality recursively.  Since $kR^+_\ell - \hat{x} < k\lambda$, for any $\ell \geq 2$ we have
\begin{align}\label{ssm-con-rate}
kR^+_\ell - \hat{x}\leq k\lambda |f'(\hat{x})|^{\ell - 2}.
\end{align}
To bound $R^-_\ell$ we apply the mean value theorem again. There exists a $z\in[\hat{x},kR^+_\ell]$ such that
\begin{align*}
\hat{x} - kR^-_\ell = f(\hat{x}) - f(kR^+_{\ell-1}) = |f'(z)|(kR^+_{\ell-1} - \hat{x}).
\end{align*}
Since $|f'(z)| \leq kd\lambda$ for all $z > 0$, combined with~\eqref{ssm-con-rate} we have
\begin{align*}
	\hat{x} - kR^-_\ell \leq kd\lambda(kR^+_{\ell-1} - \hat{x}) \leq k^2d\lambda^2|f'(\hat{x})|^{\ell - 3}.
\end{align*}
At last, $R^+_\ell - R^-_\ell = \frac{1}{k}(kR^+_\ell - \hat{x} + \hat{x} - kR^-_\ell)\leq C'_1|f'(\hat{x})|^{\ell - 3}$ for some $C_1'>0$ depending only on $d,k$ and $\lambda$.
This only gives us the desired decay rate at the $(k+1)$-uniform $d$-ary hypertree $\VHtree{k}{d}$. Move to the $(k+1)$-uniform $(d+1)$-regular hypertree $\RHtree{k}{d}$. The only difference is that the root has $d+1$ children instead of $d$.  
By the mean value theorem, this will multiply at most a finite constant factor $C_1''$ to the gap $R^+_\ell-R^-_\ell$ at the root of $\RHtree{k}{d}$, where $C_1''>0$ depends only on $d,k$ and $\lambda$. Overall, this gives us that
\[
p^+_\ell-p^-_\ell\le R^+_\ell-R^-_\ell\le C_1|f'(\hat{x})|^{\ell - 4}
\]
for some $C_1>0$ depending only on $d,k$ and $\lambda$. This finishes the case that $\lambda<\lambda_c$.

Now we consider the critical case that $\lambda = \lambda_c=\frac{d^d}{k(d-1)^{d+1}}$.  We still start by considering the $(k+1)$-uniform $d$-ary hypertree $\VHtree{k}{d}$.
It is easy to verify that in this case $\hat{x} = \frac{1}{d-1}$, $\alpha(\hat{x}) = f'(\hat{x})^2 = 1$, $z^*=\hat{x}$ and $\alpha'(\hat{x}) = 0$, where $\alpha(z)$ and $z^*$ are defined in the proof of Lemma~\ref{lemma-rate}. And we have 
	$\alpha''(\hat{x}) 
	= -\frac{(d+1)(d-1)^3}{d^2}$.
By Taylor's expansion for $g(x)$ at the fixed point $x=\hat{x}$, we have that for any constant $c > 0$ there exists a constant $x_0 > \hat{x}$ such that for any $\hat{x} < x < x_0$, it holds that
\begin{align*}
	g(x) &= g(\hat{x}) + \alpha(\hat{x})(x-\hat{x}) + \frac{\alpha'(\hat{x})}{2}(x-\hat{x})^2 + \frac{\alpha''(\hat{x})}{6}(x-\hat{x})^3 + o\left((x-\hat{x})^3\right) \\
	       &\leq \frac{1}{d-1} + x - \hat{x} -\frac{(d+1)(d-1)^3}{6d^2}(x-\hat{x})^3 + c(x-\hat{x})^3.
\end{align*}
We define a sequence $x_1 = kR^+_1, x_3 = kR^+_3 = g(x_1), \dots$ and generally $x_{2t+1}=g(x_{2t-1})$. The sequence is strictly decreasing because $R^+_\ell$ is decreasing in $\ell$. Furthermore, $\lim_{t \to \infty}x_{2t+1}=\hat{x}$. This is due to $\alpha(x) < 1$ for any $x > \hat{x}$. 	
	
	Denote $\epsilon_{2t+1} \defeq x_{2t+1} - \hat{x}$. Let $c$ be some positive constant such that $\frac{(d+1)(d-1)^3}{6d^2} - c > 0$. Denote $\beta = \frac{(d+1)(d-1)^3}{6d^2} - c$ and $\gamma = \sqrt{\frac{1}{2\beta}}$. There must be some sufficiently large $t_0$ such that $\epsilon_{2t_0+1} \leq \frac{\gamma}{\sqrt{2}}$ and for any $t > t_0$, it holds that
\begin{align*}
	\epsilon_{2t+3} = g(x_{2t+1}) - \frac{1}{d-1} \leq \epsilon_{2t+1} -\beta\epsilon_{2t+1}^3.
\end{align*}
	We apply an induction to complete the proof. For the basis, when $t = t_0$ we have $\epsilon_{2t_0+1} \leq \frac{\gamma}{\sqrt{2}}$. Assume the hypothesis 
\begin{align}\label{eq:c-wsm-hy}
	\epsilon_{2t+1} \leq \frac{\gamma}{\sqrt{t - t_0 + 2}}
\end{align}
for some $t \geq t_0$ and we will prove it holds for $t+1$. First, notice that $h(x) \defeq x -\beta x^3$ is strictly increasing when $0 \leq x \leq \frac{\gamma}{\sqrt{2}}$. Thus, we have
\begin{align*}
	\epsilon_{2t+3} &\leq \epsilon_{2t+1} -\beta\epsilon_{2t+1}^3 
	\leq \frac{\gamma}{\sqrt{t - t_0 + 2}} - \beta\frac{\gamma^3}{(t-t_0+2)^{\frac{3}{2}}}.
\end{align*}
We only need to prove that 
	$\frac{\gamma}{\sqrt{t - t_0 + 2}} - \beta\frac{\gamma^3}{(t-t_0+2)^{\frac{3}{2}}} \leq \frac{\gamma}{\sqrt{t-t_0+3}}$.
Let $t' \defeq t - t_0 + 2$. It is equivalently to show that 
\begin{align}\label{eq:c-wsm}
	t'^{\frac{3}{2}}\left(\frac{1}{\sqrt{t'}}-\frac{1}{\sqrt{t'+1}}\right) &\leq \beta\gamma^2.
\end{align}
Note that
\begin{align*}
	t'^{\frac{3}{2}}\left(\frac{1}{\sqrt{t'}}-\frac{1}{\sqrt{t'+1}}\right) &= t'^{\frac{3}{2}}\left(\frac{\sqrt{t'+1} - \sqrt{t'}}{\sqrt{t'(t'+1)}}\right) \\
	&\leq \sqrt{t'}(\sqrt{t'+1} - \sqrt{t'}) \\
	&\leq \sqrt{t'}\frac{1}{\sqrt{t'+1} + \sqrt{t'}} \\
	&\leq \frac{1}{2}.
\end{align*}
Since $\beta\gamma^2 = \frac{1}{2}$, we just prove the inequality~\eqref{eq:c-wsm}, and finishes the induction~\eqref{eq:c-wsm-hy} for all $t \geq t_0$.
In conclusion, for any $t \geq t_0$, it holds that
\begin{align*}
	kR^+_{2t+1} - \hat{x} \leq \frac{\gamma}{\sqrt{t - t_0 + 2}}.
\end{align*}
The rest of the proof is exactly the same as our proof of the case $\lambda < \lambda_c$. 
\end{proof}

\section{Approximation algorithms and inapproximability}
\label{sec:algorithms}

For $0<\varepsilon<1$, a value $\hat{Z}$ is an \concept{$\varepsilon$-approximation} of $Z$ if $(1-\epsilon)Z\le\hat{Z}\le(1+\epsilon)Z$. Recall that $\hat{x}$ is the unique fixed point solution to $\hat{x}=f_{d,k}(\hat{x})=k\lambda(1+\hat{x})^{-d}$.

\begin{theorem}\label{theorem-FPTAS}
If $\lambda<\lambda_{c}=\frac{d^d}{k(d-1)^{d+1}}$, then there exists an algorithm such that given any $\varepsilon>0$, and any hypergraph $\hyper{H}$ of $n$ vertices, of maximum edge-size at most $(k+1)$ and maximum degree at most $(d+1)$, the algorithm returns an $\varepsilon$-approximation of the partition function for the independent sets of $\hyper{H}$ with activity $\lambda$, within running time $\left(\frac{n}{\varepsilon}\right)^{O\left(\frac{1}{\kappa}\ln kd\right)}$, where $\kappa=\ln{\left(\frac{1+\hat{x}}{d\hat{x}}\right)}$. 

For the critical case where $\lambda=\lambda_{c}$, there exists an algorithm that for the above $\hyper{H}$ returns an $\varepsilon$-approximation of the log-partition function within running time $n(kd)^{O\left(\left(\frac{1}{\varepsilon}\ln\frac{1}{\varepsilon}\right)^2\right)}$.
\end{theorem}

By duality, the same algorithm with the same approximation ratio and running time works for the matchings of hypergraphs of maximum edge size at most $(d+1)$ and maximum degree at most $(k+1)$.
By Proposition~\ref{proposition-uniqueness}, $|f'_{d,k}(\hat{x})|=\frac{d\hat{x}}{1+\hat{x}}<1$ if $\lambda<\lambda_c$, therefore, when $\lambda<\lambda_c$, the running time of the algorithm is $\mathrm{Poly}(n,\frac{1}{\epsilon})$ for any bounded $k$ and $d$,  so the algorithm is an FPTAS for the partition function. And when $\lambda=\lambda_c$, the algorithm is a PTAS for the log-partition function.
The algorithmic part of the main theorem Theorem~\ref{theorem-main} is proved.

In particular, when $d=1$, the model becomes matchings of graphs of maximum degree $(k+1)$, and the uniqueness condition $\lambda<\lambda_c(\RHtree{k}{d})$ is always satisfied even for unbounded $k$ since $\lambda_c(\RHtree{k}{1})=\infty$. In this case, the fixed point $\hat{x}$ for $f_{1,k}(x)=\frac{k\lambda}{1+x}$ can be explicitly solved as $\hat{x}=\frac{-1+\sqrt{1+4k\lambda}}{2}$. We have the following corollary for matchings of graphs with unbounded maximum degree, which achieves the same bound as the algorithm in~\cite{bayati2007simple}.

\begin{corollary}
There exists an algorithm which given any graph $G$ of maximum degree at most $\Delta$, and any $\epsilon>0$, returns an $\varepsilon$-approximation of the partition function for the matchings of $G$ with activity $\lambda$,  within running time $\left(\frac{n}{\epsilon}\right)^{O(\sqrt{\lambda\Delta}\log \Delta)}$.
\end{corollary}

With the construction of hypergraph self-avoiding walk tree and the SSM, the algorithm follows the framework by Weitz~\cite{weitz2006counting}.
We will describe an algorithm of approximating the partition function for independent sets in hypergraphs with activity $\lambda$. Under duality this is the same as approximately counting matchings with activity $\lambda$.

By the standard self-reduction, approximately computing the partition function is reduced to approximately computing the marginal probabilities.
Let $\hyper{H}=(V,E)$ be a hypergraph and $V=\{v_1,\dots,v_n\}$. To calculate $Z=Z_\hyper{H}(\lambda)$, it suffices to calculate the probability of the emptyset $\mu(\varnothing)$ as it is exactly $1/Z$.
Let $\varnothing_i$ be the configuration on vertices $v_1$ up to $v_i$ where all of them are unoccupied, and $p_{v_i}^{\varnothing_{i-1}}$ the probability of $v_i$ being occupied conditioning on all vertices $v_1$ up to $v_{i-1}$ being unoccupied.
Then we have $1/Z=\prod_{i=1}^n(1- p_{v_i}^{\varnothing_{i-1}})$ and $\log Z=-\sum_{i=1}^n\log(1-p_{v_i}^{\varnothing_{i-1}})$. Note that $(1-p_{v_i}^{\varnothing_{i-1}})\ge\frac{1}{1+\lambda}$ for the probability of vertex unoccupied by an independent set and $\lambda_c\le 4$ for any $d\ge 2$ and $k\ge 1$. To get an $\varepsilon$-approximation of $Z$, it suffices to approximate each of $p_{v_i}^{\varnothing_{i-1}}$ within an additive error $\frac{\varepsilon}{2(1+\lambda)n}$. And to get an $\varepsilon$-approximation of $\log Z$, which can be obtained by getting an $\varepsilon$-approximation of every $-\log(1-p_{v_i}^{\varnothing_{i-1}})$, it is sufficient to approximate each of $p_{v_i}^{\varnothing_{i-1}}$ within an additive error $\Theta\left(\frac{\varepsilon}{\ln\frac{1}{\varepsilon}}\right)$.

By Theorem~\ref{thm-saw}, we have $p_{v}^{\sigma}=\ptree{\hyper{T}}{\sigma}$ where $\hyper{T} = \TSAW(\hyper{H}, v)$, i.e.~the marginal probability of $v$ being occupied is preserved in the SAW tree of $\hyper{H}$ expanded at $v$. And the value of $\ptree{\hyper{T}}{\sigma}$ can be computed by the tree recursion~\eqref{eq:tree-recursion}. To make the algorithm efficient we can run this recursion up to depth $t$ and assume initial value 0 for the variables at depth $t$ as the vertices they represent being unoccupied. The overall running time of the algorithm is clearly $O(n(kd)^t)$ where $t$ is the depth of the recursion. By the strong spatial mixing guaranteed by Theorem~\ref{theorem-srj-ssm-IS} and Theorem~\ref{theorem-wsm-IS}, if $\lambda<\lambda_c$, then the additive error of such estimation of $p_{v}^{\sigma}$ is bounded by $C_1\cdot\left(\frac{d\hat{x}}{1+\hat{x}}\right)^{t-4}$ for some constant $C_1>0$ depending only on $k,d$ and $\lambda$. We shall choose an integer $t$ so that $C_1\cdot\left(\frac{d\hat{x}}{1+\hat{x}}\right)^{t-4}\leq \frac{\varepsilon}{2(1+\lambda)n}$, which gives us the suitable time complexity required by the FPTAS for the partition function. 
And when $\lambda=\lambda_c$, the additive error of $p_{v}^{\sigma}$ is bounded by $C_2/\sqrt{t-t_0}$ for some constants $C_2,t_0>0$ depending only on $k,d$. We shall choose an integer $t=O((\frac{1}{\varepsilon}\ln\frac{1}{\varepsilon})^2)$ to get the desirable additive error for every marginal probability, which gives us the PTAS for the log-partition function.
This completes the proof of Theorem~\ref{theorem-FPTAS}.

\paragraph*{Inapproximability.}

For the inapproximability, by applying an AP-reduction~\cite{bordewich2008path} from the inapproximability of the hardcore model~\cite{sly2014counting,galanis2012inapproximability}, we have the following theorem. 

\begin{theorem}\label{theorem-inapprox}
If $\lambda > \frac{2k+1+(-1)^k}{k+1}\lambda_c$, there is no PRAS for the partition function or log-partition function of 
independent sets of hypergraphs with maximum degree at most $d+1$, maximum edge-size at most $k+1$ and activity $\lambda$, unless NP=RP.
\end{theorem}

\begin{proof}
The reduction is as described in~\cite{bordewich2008path}, which is reduced from the hardcore model.
Given a graph $G(V,E)$ with maximum degree at most $(d+1)$, we construct a hypergraph $\hyper{H}(V_{\hyper{H}},E_{\hyper{H}})$ as follows. For each $v\in V$, we create $t=\floor{\frac{k+1}{2}}$ distinct vertices $w_{v,1},w_{v,2},\ldots,w_{v,t}$ and let $V_{\hyper{H}}=\{w_{v,i}\mid v\in V,1\le i\le t\}$. And for every edge $e=(u,v)\in E$, we create a hyperedge $S_{e}=\{w_{u,1},\ldots,w_{u,t}, w_{v,1},\ldots,w_{v,t}\}$ and let $E_{\hyper{H}}=\{S_e\mid e\in E\}$. Clearly, the maximum degree of $\hyper{H}$ is at most $d+1$ and the maximum edge-size of $\hyper{H}$ is at most $2t\le k+1$.
We define
\[
Z_{\hyper{H}}(\lambda)=\sum_{I\text{: IS of }\hyper{H}}\lambda^{|I|}
\quad
\text{and}
\quad
Z_{G}(\lambda)=\sum_{I\text{: IS of }G}\lambda^{|I|}.
\]
Note that by the above reduction every independent set $I$ of $\group{G}$ is naturally identified to $t^{|I|}$ distinct independent sets of hypergraph $\hyper{H}$ such that a $v\in V$ is occupied by $I$ if and only if one of $w_{v,i}$ is occupied by the corresponding independent set of $\hyper{H}$. Thus $Z_{\hyper{H}}(\lambda)=Z_{G}(\lambda')$ where $\lambda'=t\lambda$.

Recall that $G$ is an arbitrary graph of maximum degree at most $d+1$. According to Sly and Sun~\cite{sly2014counting}, when $\lambda'>\frac{d^d}{(d-1)^{d+1}}$, there exists a constant $c$ such that unless NP=RP, the partition function $Z_{G}(\lambda')$ can not be approximated within a factor of $c^{n}$ in polynomial time, which means there is no PRAS for the log-partition function $\log Z_{G}(\lambda')$ when $\lambda'>\frac{d^d}{(d-1)^{d+1}}$, i.e.~when $\lambda>\frac{d^d}{\floor{(k+1)/2}(d-1)^{d+1}}=\frac{2k+1+(-1)^k}{k+1}\lambda_c$.

\end{proof}

\begin{figure}
\centering
\includegraphics[scale=.5]{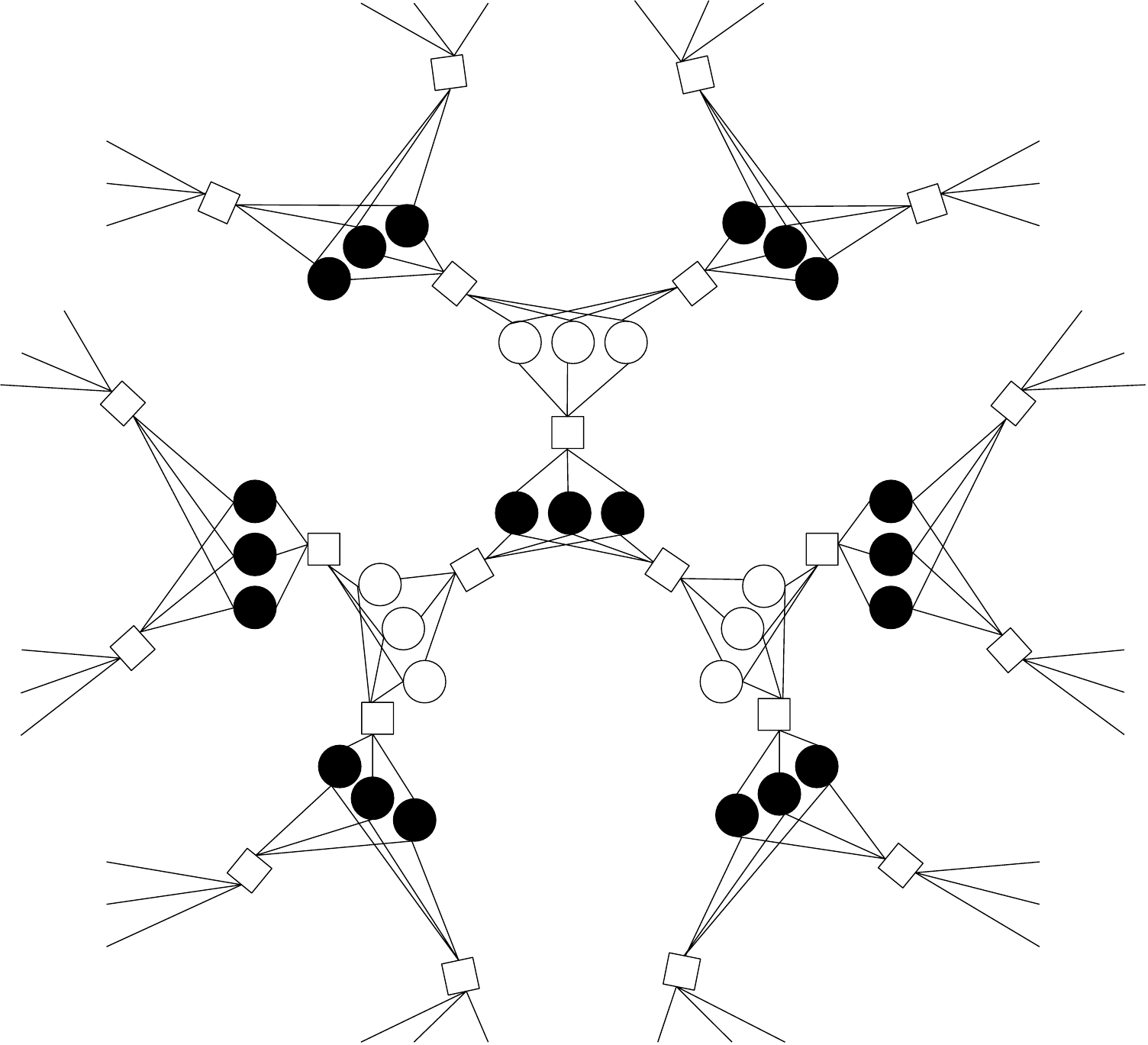}
\caption{The infinite hypergraph that achieves the uniqueness threshold $\frac{2k+1+(-1)^k}{k+1}\lambda_c$.}
\label{fig:gadget}
\end{figure}

The reduction in Theorem~\ref{theorem-inapprox} transforms a hardcore model on a graph with maximum degree $d+1$ and activity $\frac{2k+1+(-1)^k}{k+1}\lambda$ to an instance of hypergraph independent sets with maximum degree at most $d+1$, maximum edge-size at most $k+1$, and activity $\lambda$. In particular, it transforms the infinite $(d+1)$-regular tree $\RHtree{d}{1}$ to the infinite $2\lfloor(k+1)/2\rfloor$-uniform hypergraph as shown in Figure~\ref{fig:gadget}. This infinite hypergraph has the uniqueness threshold $\frac{d^d}{\floor{(k+1)/2}(d-1)^{d+1}}=\frac{2k+1+(-1)^k}{k+1}\lambda_c$.

\newcommand{\seqc}[1]{C_{{#1}}}
\newcommand{\seqpr}[2]{P_{#2}^{{#1}}}

\section{Local convergence of hypergraphs}
\label{sec:local-converge}

For the infinite $(k+1)$-uniform $(d+1)$-regular hypertree $\RHtree{k}{d}$, a group $\group{G}$ of automorphisms on  $\RHtree{k}{d}$ classifies the vertices and hyperedges in $\RHtree{k}{d}$ into orbits (equivalent classes). We consider only $\group{G}$ with finitely many orbits. By Proposition~\ref{prop-bracnching-matrix}, such group $\group{G}$ can be uniquely identified by a pair of branching matrices $(\mat{D}, \mat{K})$ defined in Section~\ref{sec:uniqueness} that classifies vertices and hyperedges in $\RHtree{k}{d}$ into finitely many types (labels), where the incidence relation between vertices and hyperedges with each type is specified by $(\mat{D}, \mat{K})$. We use $\RHtree{k}{d}^{\group{G}}$ to denote this resulting labeled hypertree.

For a finite hypergraph $\hyper{H} = (V,E)$, we also consider the classification of vertices $V=\biguplus_{i\in[\tau_v]} V_{i}$ and hyperedges $E=\biguplus_{j\in[\tau_e]} E_{j}$ into disjoint types.

Given a hypergraph $\hyper{H}$ and a vertex $v$ in  $\hyper{H}$, write $\neigh{}{t}(v)=\neigh{\hyper{H},}{t}(v)$ for the \concept{$t$-neighborhood} around $v$ in $\hyper{H}$, that is, the sub-hypergraph induced by the vertices in $\hyper{H}$ at distance at most $t$ from $v$. For the labeled hypertree $\RHtree{k}{d}^{\group{G}}$, since once the type of the root is fixed the neighborhoods are identical (in terms of types), for each $i\in[\tau_v]$, we can denote $\RHtree{k}{d}^{\group{G}}(t,i)=\neigh{T,}{t}(v)$ where $T=\RHtree{k}{d}^{\group{G}}$ and $v$ is any vertex in $T$ of type-$i$.

The following definition is inspired by those of~\cite{sly2014counting} and~\cite{dembo2010ising} for spin systems. Intuitively, a sequence of finite structures locally resemble the infinite tree structure along with the suitable symmetry which exhibits the uniqueness/nonuniqueness phase transition at the critical threshold, so the measures on the sequence of finite structures may have \concept{local weak convergence} to that on the infinite tree. The existence of such local convergence profoundly leads to several most important phase-transition-based inapproximability results~\cite{dyer2002counting, mossel2009hardness, sly2010computational, sly2014counting, galanis2012improved, galanis2012inapproximability, galanis2014inapproximability} and is a key to the success of random regular bipartite graph as a gadget for anti-ferromagnetic spin systems.

\begin{definition}[local convergence]\label{definition-local-convergence}

Let $\hyper{H}_n = (V_n,E_n)$ be a sequence of random finite hypergraphs, whose vertices $V_n=\biguplus_{i\in[\tau_v]} V_{n,i}$ and hyperedges $E_n=\biguplus_{j\in[\tau_e]} E_{n,j}$ are classified into disjoint types, and for each $i\in[\tau_v]$, let $I_{n,i}\in V_{n,s}$ denote a uniformly random vertex in $V_n$ of type-$i$.

We say the $\hyper{H}_n$ \concept{converge locally} to $\RHtree{k}{d}^{\group{G}}$, and write $\hyper{H}_n\localto\RHtree{k}{d}^{\group{G}}$, if
for all $t\ge 0$ and $i\in[\tau_v]$, $\neigh{}{t}(I_{n,i})$ converges to $\RHtree{k}{d}^{\group{G}}(t,i)$ in distribution with respect to the joint law $\law{P}_n$ of $(\hyper{H}_n, I_{n,i})$: that is,
\[
\lim_{n \to \infty} \law{P}_n\left( \neigh{}{t}(I_{n,i}) \cong \RHtree{k}{d}^{\group{G}}(t,i) \right) = 1,
\]
where $\cong$ denotes isomorphism which preserves vertex- and hyperedge-types and the incidence relation.
\end{definition}

Consider the natural uniform random walk on the incidence graph of $\RHtree{k}{d}^{\group{G}}$, and its projection onto the finitely many disjoint orbits (types) for vertices and hyperedges, which gives a (bipartite) finite Markov chain. It is quite amazing to see that the reversibility of this projected chain determines whether there exists a sequence of finite hypergraphs that  converge locally to $\RHtree{k}{d}^{\group{G}}$.

\begin{theorem}\label{theorem-local-converge}
Let $\group{G}$ be an automorphism group of $\RHtree{k}{d}$ with finitely many orbits for vertices and hyperedges. 
Let $\mat{D}$ and $\mat{K}$ be the branching matrices that corresponds to $\group{G}$ as defined in Section~\ref{sec:uniqueness}. 
There is a sequence of random finite hypergraphs $\hyper{H}_n\localto\RHtree{k}{d}^{\group{G}}$ if and only if the  Markov chain $\boldsymbol{P}=\begin{bmatrix}
\boldsymbol{0} & \frac{1}{d+1}\boldsymbol{D}\\
\frac{1}{k+1}\boldsymbol{K} & \boldsymbol{0}
\end{bmatrix}$ is time-reversible.
\end{theorem}

We say a uniform random walk over a hypergraph $\hyper{H}$ is a uniform random walk on the incidence graph of $\hyper{H}$: that is, a random walk moves between vertices and hyperedges.
Then the Markov chain $\boldsymbol{P}$ is the projection of the uniform random walk over $\RHtree{k}{d}$ onto the equivalent classes of vertices and hyperedges (i.e.~the orbits of the automorphism group $\group{G}$ that corresponds to the $\mat{D}$ and $\mat{K}$).
Meanwhile, matrix $\begin{bmatrix}
\boldsymbol{0} & \boldsymbol{D}\\
\boldsymbol{K} & \boldsymbol{0}
\end{bmatrix}$ is the adjacent matrix for a directed bipartite graph that describes the (weighted) incidence relation between vertex- and hyperedge-types in the following way: each directed bipartite edge from vertex-type-$i$ to hyperedge-type-$j$ (or vice versa) is assigned with weight $d_{ij}$ (or $k_{ji}$). So the Markov chain $\boldsymbol{P}$ is also the random walk on this directed bipartite graph where the transition probability of each directed edge is proportional to its weight.

For the bipartite Markov chain $\boldsymbol{P}$, recall that due to Proposition~\ref{prop-bracnching-matrix}, $\boldsymbol{P}$ must be irreducible. Then the time-reversibility of $\boldsymbol{P}$ is equivalent to the following:
There exist positive vectors $\vec{p}=(p_i)_{i\in[\tau_v]}$ and $\vec{q}=(q_j)_{j\in[\tau_e]}$ that satisfy the bipartite {detailed balanced equation}:
\[
p_i d_{ij} = q_j k_{ji}
\]
for every $(i,j)\in[\tau_v]\times[\tau_e]$. Without loss of generality, we assume $\sum_ip_i + \sum_jq_j = 1$.

In fact, it is easy to check that  $\vec{p}\mat{D}=(k+1)\vec{q}$ and $\vec{q}\mat{K}=(d+1)\vec{p}$, therefore the $\vec{p}$ and $\vec{q}$ are respectively the left eigenvector of $\mat{D}\mat{K}$ and $\mat{K}\mat{D}$ both with eigenvalue $(d+1)(k+1)$. Since both $\mat{D}\mat{K}$ and $\mat{K}\mat{D}$ are irreducible, due to the Perron-Frobenius theorem, the only positive left eigenvectors $\vec{p}$ and $\vec{q}$ are the ones that are  associated with the Perron-Frobenius eigenvalue $(d+1)(k+1)$ and are one-dimensional. 

Furthermore, it must holds that $\frac{||\vec{p}||_1}{||\vec{q}||_1} = \frac{k+1}{d+1}$. Denote $\vec{p'} = \frac{\vec{p}}{\|\vec{p}\|_1}$ and $\vec{q'} = \frac{\vec{q}}{\|\vec{q}\|_1}$. We have $\|\vec{p'}\|=\|\vec{q'}\|=1$ and
$p'_i \frac{d_{ij}}{d+1} = q'_j \frac{k_{ji}}{k+1}$
for every $(i,j)\in[\tau_v]\times[\tau_e]$, i.e.~$\vec{p'}$ is the $\concept{vertex-stationary distribution}$ and $\vec{q'}$ is the $\concept{edge-stationary distribution}$. 
We will mostly use $\vec{p}$ and $\vec{q}$ in our proof of Theorem~\ref{theorem-local-converge}. 
Recall for the automorphism group $\widehat{\group{G}}$ defined in Section~\ref{sec:uniqueness} such that $\lambda_c(\RHtree{k}{d}^{\widehat{\group{G}}})=\lambda_c(\RHtree{k}{d})=\frac{d^d}{k(d-1)^{d+1}}$, i.e.~the uniqueness of $\widehat{\group{G}}$-translation-invariant Gibbs measures on $\RHtree{k}{d}$ represents the uniqueness of all Gibbs measures on $\RHtree{k}{d}$, the branching matrices are given as
$\widehat{\boldsymbol{D}}=\begin{bmatrix}
1 & d\\
d & 1
\end{bmatrix}$ and $\widehat{\boldsymbol{K}}=\begin{bmatrix}
k & 1\\
1 & k
\end{bmatrix}$.
It is easy to verify that the resulting Markov chain $\widehat{\boldsymbol{P}}$ is not time-reversible. It then follows from Theorem~\ref{theorem-local-converge} that there does not exist \emph{any} sequence of random finite hypergraphs that converge locally to $\RHtree{k}{d}$ with the symmetry $\widehat{\group{G}}$ assumed by the extremal Gibbs measures $\mu^+,\mu^-$ whose uniqueness represents the uniqueness of all Gibbs measures.

\begin{remark*}
Given branching matrices $\mat{D}$ and $\mat{K}$, instead of considering $\hyper{H}_n$ that converges locally for every type to the $\RHtree{k}{d}^{\group{G}}$  as in Definition~\ref{definition-local-convergence}, we can alternatively define a sequence $\hyper{H}_n$ that \concept{converges locally in average} to the $\RHtree{k}{d}^{\group{G}}$: that is, for all $t>0$, the $\neigh{}{t}(I_n)$ converges to $\RHtree{k}{d}^{\group{G}}(t,I)$ in distribution, where $I_n$ denotes a uniformly random vertex in the finite hypergraph $\hyper{H}_n$, and $I$ denotes a random vertex-type chosen according to the vertex-stationary distribution $\vec{p}'$. This definition looks more analogous to the local convergence defined in~\cite{sly2014counting} for the anti-ferromagnetic 2-spin system. But we will see the two definitions are equivalent: A sequence $\hyper{H}_n\localto\RHtree{k}{d}^{\group{G}}$ also converges locally to $\RHtree{k}{d}^{\group{G}}$ in average, since by double counting the portion of vertices of type-$i$ must converge to $p_i'$ as $n \to \infty$; and conversely, a sequence converges locally to $\RHtree{k}{d}^{\group{G}}$ in average must also have$\hyper{H}_n\localto\RHtree{k}{d}^{\group{G}}$, simply because neighborhoods of vertices of different types cannot be isomorphic to each other. 
\end{remark*}

\begin{proof}[Proof of Theorem~\ref{theorem-local-converge}]
We will prove the necessity of the reversibility of the chain by a double counting argument and the sufficiency is proved by explicitly constructing the sequence of the finite hypergraphs.

\paragraph*{Double counting.}
Let $\hyper{H}_n=(V_n, E_n)$ where $V_n=\biguplus_{s \in \tau_v}V_{n,s}$ and $E_n=\biguplus_{t \in \tau_e}E_{n,t}$. Assume that $\hyper{H}_n\localto\RHtree{k}{d}^{\group{G}}$.

For $\RHtree{k}{d}^{\group{G}}$ such that there is a hypergraph sequence $\hyper{H}_n=(V_n=\biguplus_{s \in \tau_v}V_{n,s}$, $E_n=\biguplus_{t \in \tau_e}E_{n,t})$ converging locally to $\RHtree{k}{d}^{\group{G}}$, we show that the Markov chain $\mat{P}$ is time reversible. The proof is by a double counting of the number of vertex-hyperedge pairs with specific type combination.

Since the $1$-neighborhood of the vertex with each type in $\hyper{H}_n$ converges in distribution to the $1$-neighborhood of the vertex with the same type in $\RHtree{k}{d}^{\group{G}}$, for sufficiently large $n$, we have all but a $o(1)$-fraction of vertices in $\hyper{H}_n$ whose local transitions between vertex-types and hyperedge-types within 1-step are given precisely by $\mat{D}$ and $\mat{K}$. Thus, for every $(i,j) \in [\tau_v] \times [\tau_e]$, the total number of incident vertex-hyperedge pair $(v,e)$ with $v \in V_{n,i}$ and $e \in E_{n,j}$ (counted from the vertex-side and from the hyperedge-side) is given by
\[
d_{ij} (|V_{n,i}| + o(1)) = k_{ji} (|E_{n,j}| + o(1)).
\]
As $n \to \infty$, we will have $(d_{ij} |V_{n,i}|)/( k_{ji}|E_{n,j}|) \to 1$, or equivalently
\[
\frac{|E_{n,j}|}{ |V_{n,i}|} \to \frac{d_{ij}}{ k_{ji}}
\]
for all $(i,j) \in [\tau_v] \times [\tau_e]$ such that $d_{ij}, k_{ji} \ne 0$.
Thus there exists positive $p_i,q_j$ such that $q_j / p_i = d_{ij} / k_{ji}$ for all such $(i,j)$. Since $\mat{D}\mat{K}$ and $\mat{K}\mat{D}$ are irreducible, we have unique corresponding positive left eigenvectors, which is $(p_i)_{i \in [\tau_v]}, (q_j)_{j \in [\tau_e]}$ here, such that $p_i d_{ij} = q_j k_{ji}$ for all $(i,j)$.

\paragraph*{Construction of $\hyper{H}_n$.}
Assume the Markov chain $\boldsymbol{P} $ in Theorem~\ref{theorem-local-converge} to be time-reversible, and let $\vec{p}=(p_i)_{i\in[\tau_v]}$ and $\vec{q}=(q_j)_{j\in[\tau_e]}$ be the unique positive vectors satisfying $p_i d_{ij} = q_j k_{ji}$ for every $(i,j)\in[\tau_v]\times[\tau_e]$ and $\sum_ip_i + \sum_jq_j = 1$.
The sequence of finite hypergraph sequence $\hyper{H}_n$ that converges locally to $\RHtree{\mat{K}}{\mat{D}}$ is constructed as follows. The number $n$ is approximately the total number of vertices and hyperedges in $\hyper{H}_n$ (where the approximation is due to rounding).
\begin{itemize}
\item For each $s\in[\tau_v]$ and $t\in[\tau_e]$, let $V_{n,s}$ be the set of $\ceil{p_s n}$ vertices of type $s$, and $E_{n,t}$ be the set of $\ceil{q_t n}$ hyperedges of type $t$.
We then describe hypergraphs $\hyper{H}_n=(V_n,E_n)$ where $V_n=\biguplus_{s \in \tau_v}V_{n,s}$ and $E_n=\biguplus_{t \in \tau_e}E_{n,t}$.
\item For each $s\in[\tau_v]$ and $t\in[\tau_e]$, let $N_{s,t} \defeq \ceil{d_{st} p_s n} = \ceil{k_{ts} q_t n}$. Sample a uniformly random permutation $f : [N_{s,t}] \to [N_{s,t}]$, and create an incidence between the $i$-th vertex in $V_{n,s}$ and the $j$-th hyperedge in $E_{n,t}$ for every $(a,b=f(a))$ with $a \in i + \ceil{p_s n} \mathbb{Z}$ and $b \in j + \ceil{q_t n} \mathbb{Z}$.
\end{itemize}

Note that as normalized Perron eigenvectors for irreducible integer matrices, the $\vec{p}$ and $\vec{q}$ must be rational. Then there are infinitely many $n$ such that $N_{s,t}/|V_{n,s}| = d_{st}$ and $N_{s,t}/|E_{n,t}| = k_{ts}$. Without loss of generality, we can consider only these $n$, since otherwise it will contribute at most $o(1)$-fractions of bad neighborhoods.

Viewing multi-edges in the incidence graph of $\hyper{H}_n$ as different edges, it holds that each vertex of type-$s$ is incident to exactly $d_{st}$ hyperedges of type-$t$ and each hyperedge of type-$t$ is incident by exactly $k_{ts}$ vertices of type-$s$. Therefore it is sufficient to show that for any finite $r>0$ the probability that the $r$-neighborhood of a vertex in $\hyper{H}_n$ has no circle is 1 as $n \to \infty$, i.e.~almost surely the $r$-neighborhood of a vertex in $\hyper{H}_n$ is a hypertree. This can be proved easily by a standard routine of Galton-Watson branching process (see e.g.~Ch.~9 in~\cite{janson2011random}) since the neighborhood is of constant size and the probability of reencountering a vertex or an edge from a population whose size goes to~$\infty$ goes to~0.

\end{proof}



\end{document}